\documentclass[10pt, showpacs, amsmath, amssymb, twocolumn, aps, prl, superscriptaddress, notitlepage, floatfix]{revtex4-2} 

\usepackage{graphicx}
\usepackage{subfigure}
\usepackage{bbm}
\usepackage{amsmath}
\usepackage{amssymb}
\usepackage{amsfonts}
\usepackage{amsthm}
\usepackage{mathrsfs}
\usepackage{bm}
\usepackage{url}
\usepackage[colorlinks=true, linkcolor=blue]{hyperref}
\usepackage[T1]{fontenc}

\usepackage{mathtools}

\usepackage[capitalise,nameinlink]{cleveref}

\usepackage{braket}
\usepackage{dcolumn}
\usepackage{color}
\usepackage{booktabs}
\usepackage{nicefrac} 
\usepackage{mathdots}
\usepackage[dvipsnames]{xcolor}

\usepackage{comment}
\usepackage[normalem]{ulem}
\newcommand\redsout{\bgroup\markoverwith{\textcolor{red}{\rule[0.5ex]{2pt}{0.4pt}}}\ULon}

\newcommand{\E}{\mathcal{E}}

\newcommand{\sgn}{\mathrm{sgn}}

\newtheorem{theorem}{Theorem}

\newtheorem{corollary}{Corollary}

\begin{document}

\title{Correlation Geometry of Quantum Sensor Networks: Local--Global Information Flow and Local Privacy}

\author{Gong-Chu Li}
\author{Lei Chen}
\author{Xu-Song Hong}
\affiliation{Laboratory of Quantum Information, University of Science and Technology of China, Hefei 230026, China.}
\affiliation{Anhui Province Key Laboratory of Quantum Network, Hefei, Anhui 230026, China.}
\affiliation{CAS Center For Excellence in Quantum Information and Quantum Physics, University of Science and Technology of China, Hefei, Anhui 230026, China.}
\affiliation{Hefei National Laboratory, Hefei 230088, China}
\author{Hua-Qing Xu}
\author{Yuancheng Liu}
\author{Si-Qi Zhang}
\author{Jia-Hao Zhao}
\affiliation{Laboratory of Quantum Information, University of Science and Technology of China, Hefei 230026, China.}
\affiliation{Anhui Province Key Laboratory of Quantum Network, Hefei, Anhui 230026, China.}
\affiliation{CAS Center For Excellence in Quantum Information and Quantum Physics, University of Science and Technology of China, Hefei, Anhui 230026, China.}

\author{Geng Chen}
\email{chengeng@ustc.edu.cn}
\affiliation{Laboratory of Quantum Information, University of Science and Technology of China, Hefei 230026, China.}
\affiliation{Anhui Province Key Laboratory of Quantum Network, Hefei, Anhui 230026, China.}
\affiliation{CAS Center For Excellence in Quantum Information and Quantum Physics, University of Science and Technology of China, Hefei, Anhui 230026, China.}
\affiliation{Hefei National Laboratory, Hefei 230088, China}

\author{Chuan-Feng Li}
\email{cfli@ustc.edu.cn}
\affiliation{Laboratory of Quantum Information, University of Science and Technology of China, Hefei 230026, China.}
\affiliation{Anhui Province Key Laboratory of Quantum Network, Hefei, Anhui 230026, China.}
\affiliation{CAS Center For Excellence in Quantum Information and Quantum Physics, University of Science and Technology of China, Hefei, Anhui 230026, China.}
\affiliation{Hefei National Laboratory, Hefei 230088, China}

\author{Guang-Can Guo}
\affiliation{Laboratory of Quantum Information, University of Science and Technology of China, Hefei 230026, China.}
\affiliation{Anhui Province Key Laboratory of Quantum Network, Hefei, Anhui 230026, China.}
\affiliation{CAS Center For Excellence in Quantum Information and Quantum Physics, University of Science and Technology of China, Hefei, Anhui 230026, China.}
\affiliation{Hefei National Laboratory, Hefei 230088, China}

\begin{abstract}
Quantum sensor networks (QSN) typically encode  $N$ unknown parameters while targeting a single linear combination, rendering the $N-1$ remaining parameters as nuisance directions. 
To rigorously quantify estimation precision under such nuisances,  
we use the effective quantum Fisher information (EQFI) and establish a ``barrel-effect'' bottleneck: the global EQFI cannot exceed the weakest weighted local sensing capacity. 
To elucidate the information allocation mechanism underlying this bottleneck, we derive an exact local--global phase map that delineates how the trade-off between local and global EQFI depends dynamically on quantum correlations, and accordingly we identify concrete conditions for saturating the bottleneck bound. 
Notably, this geometric map uncovers a counterintuitive ``overcorrelated'' regime where excessive correlations actively degrade both local and global performance.
Finally, we apply the phase map to intrinsic local privacy and identify the condition under which every local parameter is inaccessible while the desired global combination remains estimable.
Overall, our work provides a principled methodology for engineering optimal network states in quantum sensing architectures.
\end{abstract}

\maketitle

\textit{Introduction---}
Distributed quantum sensing (DQS) uses a spatially separated quantum sensor network (QSN) as illustrated in \cref{fig: demo}, each sensor encoding an unknown local parameter $\theta_i$, to estimate global properties of a field. These properties are typically weighted linear combinations of local parameters, $\theta_w=\mathbf{w}^T\boldsymbol{\theta}$~\cite{proctor2018multiparameter,eldredge2018optimal,ge2018distributed, gessner2018sensitivity, zhang2021distributed,xu2025informationally,pezze2025advances}. For independent probes, the estimation variance is bounded by the standard quantum limit (SQL), $\Delta^2\hat{\theta}_w\propto1/N$~\cite{giovannetti2004quantum,giovannetti2006quantum,degen2017quantum,pezze2018quantum,giovannetti2011advances}, whereas suitable network correlations can enable Heisenberg-limited (HL) scaling, $\Delta^2\hat{\theta}_w\propto1/N^2$~\cite{proctor2018multiparameter,eldredge2018optimal,ge2018distributed}. This capability underlies applications ranging from clock synchronization~\cite{komar2014quantum} and gravitational-wave detection~\cite{tse2019quantum} to spatially resolved field sensing~\cite{liu2021distributed,guo2020distributed}.

Previous studies have analyzed the achievable global estimation precision within a constrained one-dimensional parameter model -- namely, a setting in which only a single prescribed parameter direction varies, and the precision is fully characterized by the projected quantum Fisher information (QFI), $F_w$ \cite{proctor2018multiparameter,eldredge2018optimal,ge2018distributed,gessner2018sensitivity}. However, when all $N$ local parameters are unknown, estimating $\theta_w$ is no longer a standalone single-parameter estimation problem. Rather, it constitutes a function-estimation task embedded in an $N$-dimensional parameter space, wherein the remaining $N-1$ parameters act as nuisance parameters whose influence cannot, in general, be neglected~\cite{eldredge2018optimal,ge2018distributed,qian2021optimal,suzuki2020quantum}. 
A limited number of prior works have incorporated these nuisance directions into the attainable precision ~\cite{eldredge2018optimal,rubio2020quantum,suzuki2020quantum,tsang2020quantum}.
Nevertheless, a rigorous and systematic characterization of nuisance-aware QSNs in terms of the resulting information remains lacking.

\begin{figure}[t]
    \centering
    \includegraphics[width=\linewidth]{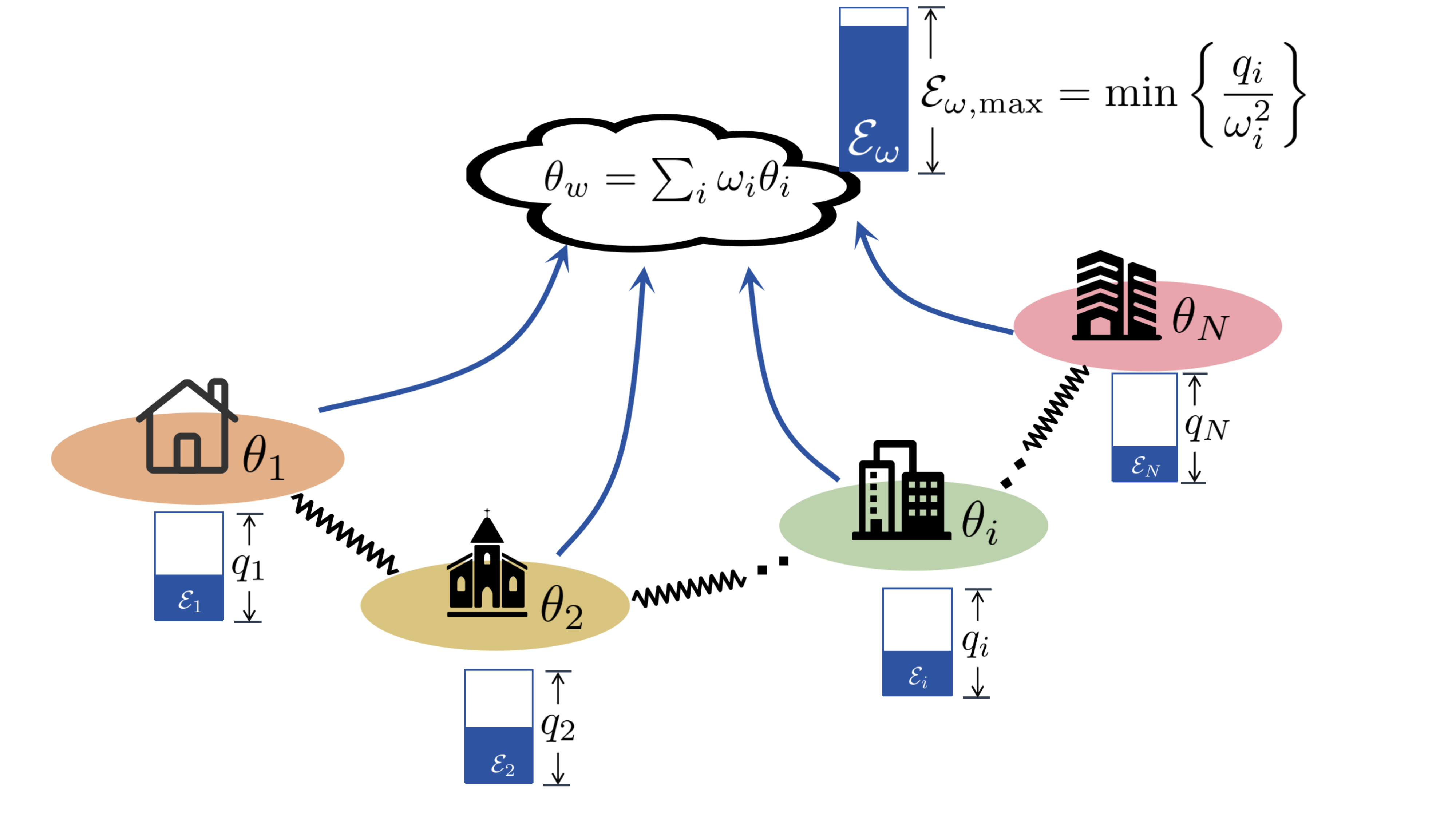}
    \caption{
    \textit{Local--global information distribution in a QSN.}
    Spatially separated sensor nodes ($\theta_1,\theta_2,\dots,\theta_N$) are connected through a multipartite quantum probe state. Finite local resources fix the diagonal sensing capacities $q_i\equiv Q_{ii}$, while the blue columns represent the operational local EQFIs $\E_i$. Network correlations redistribute metrological information between these local parameters and the global parameter $\theta_w=\sum_iw_i\theta_i$, estimated at the central node. %The global EQFI is bounded by the weakest weighted capacity, $\E_w\leq\min_{i:w_i\neq0}q_i/w_i^2$, producing a network ``barrel effect.''
    }
    \label{fig: demo}
    \vspace{-6mm}
\end{figure}

It is also well established that the ultimate achievable global estimation precision is fundamentally constrained by the physical limitations of individual sensor nodes~\cite{gessner2018sensitivity}-including, for instance, finite particle number~\cite{ge2018distributed} or limited probe energy~\cite{zhuang2018distributed} ---as well as by the correlation structure among the nodes. Thus, a central theoretical question underlying QSNs is: \textit{Given fixed local sensing capabilities, how can the attainable precision for estimating $\theta_w$ be rigorously characterized, and how should the inter-node correlation geometry be engineered to achieve this fundamental limit?}

In this paper, we formulate the nuisance-aware attainable precision in terms of the effective quantum Fisher information (EQFI), $\mathcal{E}_w$. We first establish a ``barrel-effect'' bottleneck: the global EQFI is bounded by the weakest weighted local sensing capacity. We then derive an exact local--global phase map that reveals how correlations redistribute local and global information beneath this bound.
In particular, the map distinguishes a beneficial trade-off from an \textit{overcorrelated regime}, in which correlations beyond the optimum reduce both local and global EQFIs. The phase map further yields a recursive construction of optimal correlations that saturate the bottleneck. A permutation-invariant example further reveals the sensitivity of the nuisance-aware EQFI to target-weight nonuniformity, which can drive an HL-to-SQL transition. Finally, we introduce \textit{intrinsic local privacy}, a more flexible alternative to functional privacy~\cite{shettell2022private,hassani2025privacy,bugalho2025private,ho2026quantum,farokhi2026precision}, and identify it with the phase-map boundary at which every local parameter is inaccessible while the global target remains estimable.

\textit{Effective Quantum Fisher Information---}
We consider an $N$-node QSN described by a multiparameter quantum Fisher information matrix (QFIM) $\mathbf{Q}$, with elements $Q_{ij}=\tfrac{1}{2}\operatorname{tr}[\rho\{L_i,L_j\}]$. Here, the symmetric logarithmic derivative (SLD) $L_i$ is defined by $\partial_{\theta_i}\rho=(\rho L_i+L_i\rho)/2$. The diagonal element $q_i\equiv Q_{ii}=\operatorname{tr}(\rho L_i^2)$ is the single-parameter QFI of node $i$, whereas the off-diagonal elements quantify information coupling between parameter directions. For any locally unbiased estimator $\bm{\hat{\theta}}$ based on $\nu$ independent measurements, the SLD quantum Cram\'er--Rao bound gives~\cite{helstrom1969quantum,holevo2011probabilistic,liu2020quantum}
\begin{equation}
    \operatorname{Cov}[\bm{\hat{\theta}}]
    \geq\frac{1}{\nu}\mathbf{Q}^{-1}.
\end{equation}
Although this matrix bound need not be simultaneously saturable for all parameters, we consider only one scalar target at a time and do not assume simultaneous optimal estimation. We take $\mathbf{Q}$ to be nonsingular in the following derivations; singular and limiting cases are discussed in Sec.~I of \cite{supp}.

To estimate a local parameter $\theta_i$ while treating all other parameters as unknown, we partition the QFIM as $\mathbf{Q}=\begin{pmatrix}q_i & \mathbf{c}_i^T\\ \mathbf{c}_i & \mathbf{Q}_{\not{i}}\end{pmatrix}$, where $\mathbf{c}_i$ describes the information coupling between node $i$ and the remaining subnetwork, and $\mathbf{Q}_{\not{i}}$ is the corresponding principal submatrix. The Schur complement then gives
\begin{equation}
    \Delta^2\hat{\theta}_i\geq\frac{1}{\nu\E_i},\qquad
    \E_i\equiv[(\mathbf{Q}^{-1})_{ii}]^{-1}=q_i-\mathbf{c}_i^T\mathbf{Q}_{\not{i}}^{-1}\mathbf{c}_i.
\end{equation}
We refer to $\E_i$ as the EQFI for node $i$~\cite{suzuki2020quantum,tsang2020quantum}. The subtracted term is the information penalty induced by the unknown network parameters, implying $\E_i\leq q_i$.

For the weighted global target $\theta_w=\mathbf{w}^T\bm{\theta}$, the variance of $\hat{\theta}_w=\mathbf{w}^T\bm{\hat{\theta}}$ similarly satisfies~\cite{ge2018distributed,proctor2018multiparameter,eldredge2018optimal}
\begin{equation}
    \Delta^2\hat{\theta}_w \geq\frac{1}{\nu\E_w}, \qquad
    \E_w \equiv \bigl[\mathbf{w}^T\mathbf{Q}^{-1}\mathbf{w}\bigr]^{-1}.
\end{equation}
Thus, $\E_i$ and $\E_w$ quantify the nuisance-aware information available for the local and global targets, respectively. We next determine the maximum global EQFI compatible with fixed diagonal sensing capacities $\{q_i\}$.

\textit{Resource-Constrained Global Precision---}
Finite local resources, such as particle number or probe energy~\cite{zhuang2018distributed,gessner2018sensitivity}, constrain the diagonal QFIs $q_i\equiv Q_{ii}$ of the individual nodes. 
Thus, the fundamental question posed earlier reduces to: \textit{How can the global EQFI be maximized, subject to fixed local sensing capacities $\{q_i\}$?}

To answer this question, we establish the following network-level precision ceiling.

\begin{theorem}[Fundamental Precision Bottleneck]\label{thm: maximum}
%\vspace{-2mm}
For an $N$-node QSN with fixed local capacities $\{q_i\}$ and a nonzero target weight vector $\mathbf{w}$, the global EQFI satisfies
\begin{equation}
    \E_w\leq \kappa,\qquad \kappa\equiv\min_{i:w_i\ne 0}\kappa_i,
    \qquad \kappa_i\equiv\frac{q_i}{w_i^2},
\end{equation}
where $\kappa_i$ is the intrinsic weighted capacity of node $i$ and is taken to be infinite when $w_i=0$. At the QFIM level, this bound is tight.
\vspace{-2mm}
\end{theorem}

\begin{proof}
For a nonsingular QFIM $\mathbf{Q}$, the Cauchy--Schwarz inequality gives $w_i^2=(\mathbf{e}_i^T\mathbf{w})^2 =\bigl(\mathbf{e}_i^T\mathbf{Q}^{1/2}\mathbf{Q}^{-1/2}\mathbf{w}\bigr)^2 \leq(\mathbf{e}_i^T\mathbf{Q}\mathbf{e}_i) (\mathbf{w}^T\mathbf{Q}^{-1}\mathbf{w})=q_i/\E_w$.
Thus, $\E_w\leq q_i/w_i^2$ for every $i$ with $w_i\neq0$, proving the bound. The nodewise inequality underlying this bound was previously derived in Ref.~\cite{eldredge2018optimal}. Tightness is witnessed by a QFIM $\mathbf{Q}_\star=\kappa\mathbf{w}\mathbf{w}^T+\operatorname{diag}(q_i-\kappa w_i^2)$, which has the prescribed diagonals and satisfies $\E_w=\kappa$. The recursive construction below explains how this optimal form emerges from the phase-map matching condition. More details in Sec.~II of \cite{supp}.
\vspace{-2mm}
\end{proof}

\begin{figure}[t]
\centering
\includegraphics[width=\linewidth]{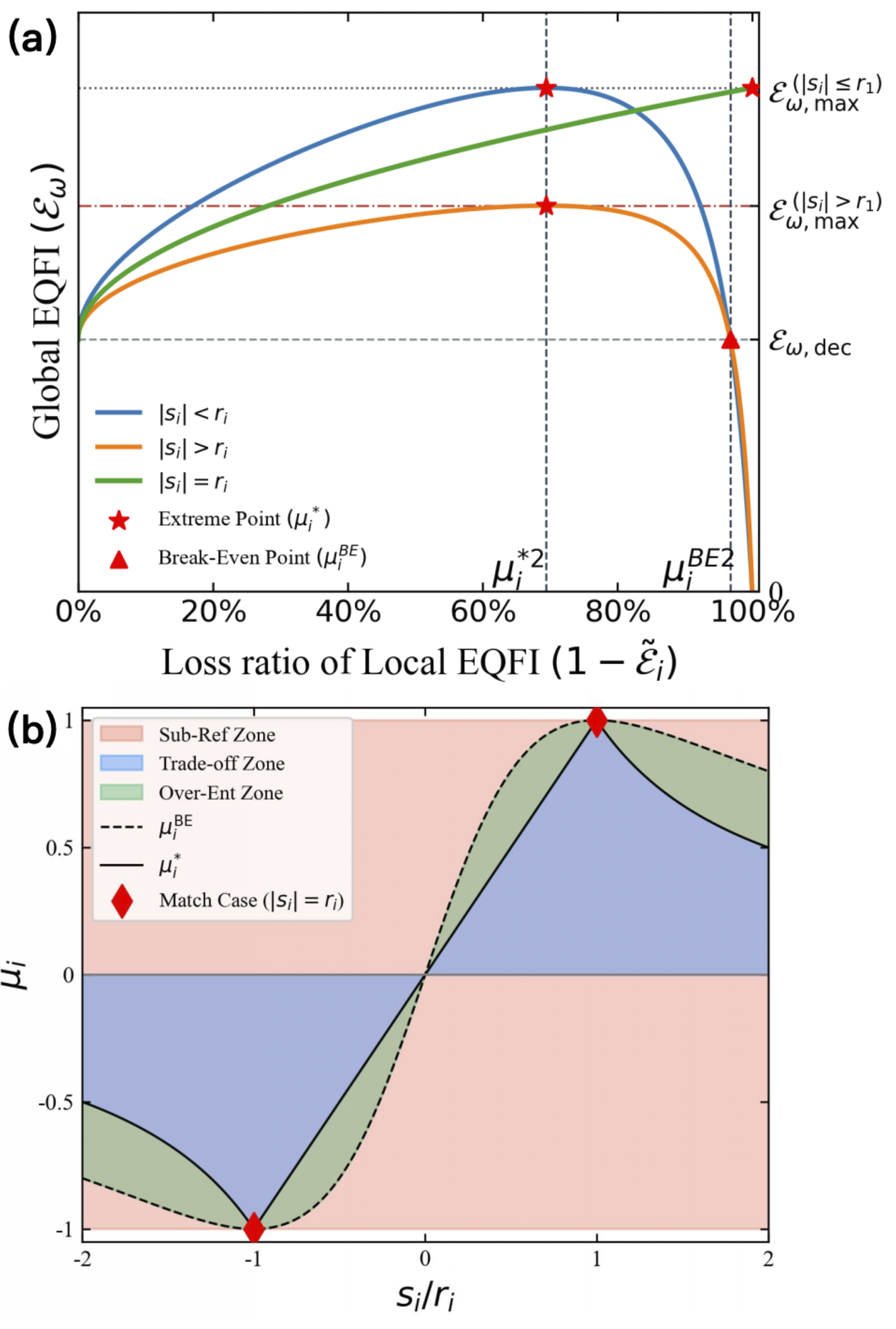}
\caption{
\textit{Local--global phase map.}
(a) Global EQFI versus the fractional local-EQFI loss, $1-\widetilde{\E}_i=\mu_i^2$, along the aligned branch $\sgn(\mu_i)=\sgn(s_i)$. Increasing the correlation enhances the global EQFI up to the optimum $\mu_i^*$ (trade-off zone). Beyond this optimum, both the local and global EQFIs decrease (overcorrelated zone). Beyond the break-even point $\mu_i^{\mathrm{BE}}$, the global EQFI falls below its decoupled value (sub-reference zone). In the matched case $|s_i|=r_i$, the overcorrelated zone vanishes, and the optimum is approached at the singular endpoint.
(b) Trade-off, overcorrelated, and sub-reference zones in the $(\mu_i,s_i/r_i)$ plane. Red diamonds mark the matched singular points, $|s_i|=r_i$ and $\mu_i=\sgn(s_i)$.
}
\label{fig:phasemap}
\end{figure}

\cref{thm: maximum} admits a straightforward ``barrel-effect'' interpretation: regardless of the network correlation structure, the global EQFI cannot exceed the smallest weighted capacity.

The theorem also determines when Heisenberg scaling is compatible with the local resource constraints. For a family of networks with $q_i=\Theta(1)$ and normalized weights $|w_i|=\Theta(1/N)$ for all participating nodes, saturation gives $\E_{w,\max}=\kappa=\Theta(N^2)$ and hence HL variance scaling, $\Delta^2\hat{\theta}_w=\Theta(N^{-2})$.

\textit{Local--Global Phase Map---}
To reveal the correlation mechanism underlying this bottleneck, we derive a phase map that determines how node--subnetwork correlations redistribute information between the local parameter and the global target.

\begin{theorem}[Node--Subnetwork Phase Map]
\label{thm:phase_map}
Consider a node $i$ with $q_i>0$, $\mathbf{Q}_{\not{i}}>0$, and $\mathbf{w}_{\not{i}}\neq0$. Write $\mathbf{c}_i=c_i\mathbf{v}_i$, where $c_i\in\mathbb{R}$ and $\|\mathbf{v}_i\|=1$, choosing the orientation of $\mathbf{v}_i$ such that $D_{\not{i}}\equiv\mathbf{w}_{\not{i}}^T\mathbf{Q}_{\not{i}}^{-1}\mathbf{v}_i\geq0$. Let $A_{\not{i}}=\mathbf{w}_{\not{i}}^T\mathbf{Q}_{\not{i}}^{-1}\mathbf{w}_{\not{i}}$ and $B_{\not{i}}=\mathbf{v}_i^T\mathbf{Q}_{\not{i}}^{-1}\mathbf{v}_i$, and define three dimensionless parameters:
\begin{equation}
    \mu_i=c_i\sqrt{\frac{B_{\not{i}}}{q_i}},~
    r_i=\frac{D_{\not{i}}}{\sqrt{A_{\not{i}}B_{\not{i}}}},~
    s_i=\frac{w_i}{\sqrt{q_iA_{\not{i}}}}.
\end{equation}
Here, $\mu_i$ is the signed normalized correlation strength and $r_i$ is the geometric overlap between the correlation and target directions. Positivity of $\mathbf{Q}$ implies $|\mu_i|\leq1$ and $0\leq r_i\leq1$. For $|\mu_i|<1$, the local and global EQFIs are
\begin{equation}\label{eq: phase-map}
    \E_i=q_i(1-\mu_i^2),\qquad
    \E_w=\frac{1}{A_{\not{i}}}\left[1+\frac{(s_i-r_i\mu_i)^2}{1-\mu_i^2}\right]^{-1}.
\end{equation}
The singular boundaries $|\mu_i|=1$ are defined by continuity.
\end{theorem}

From \cref{eq: phase-map}, the fractional loss of locally accessible information is
\begin{equation}
    1-\widetilde{\E}_i=\mu_i^2, \quad \widetilde{\E}_i\equiv{\E_i}/{q_i}. 
\end{equation}
At the decoupled point $\mu_i=0$, the global EQFI is $\E_{w,\mathrm{dec}}=\frac{1}{A_{\not{i}}(1+s_i^2)}$. If $\sgn(\mu_i)=-\sgn(s_i)$, correlations increase the mismatch term in \cref{eq: phase-map}, giving $\E_w<\E_{w,\mathrm{dec}}$. We therefore focus on the aligned branch $\sgn(\mu_i)=\sgn(s_i)$, as shown in \cref{fig:phasemap}~(a).

For nonzero $r_i$ and $s_i$, the global EQFI increases with $|\mu_i|$ until
$\mu_i^*=\sgn(s_i)\min\{|s_i|/r_i,r_i/|s_i|\}$. Before $\mu_i^*$, the global EQFI increases while the local EQFI decreases, defining the local--global trade-off zone. Beyond $\mu_i^*$, both quantities decrease, producing an overcorrelated regime in which additional correlation becomes a metrological penalty rather than a resource. The global EQFI returns to its decoupled value at $\mu_i^{\mathrm{BE}}=2r_is_i/(r_i^2+s_i^2)$.

Along the aligned branch, $|\mu_i|>|\mu_i^{\mathrm{BE}}|$ implies $\E_w<\E_{w,\mathrm{dec}}$, defining the sub-reference region. If $|s_i|\neq r_i$, the limit $\mu_i\rightarrow\operatorname{sgn}(s_i)$ drives both $\E_i$ and $\E_w$ to zero. In the matched case $|s_i|=r_i$, however, $\mu_i^*=\operatorname{sgn}(s_i)$ lies at the singular boundary: $\E_i\rightarrow0$ while $\E_w\rightarrow A_{\not{i}}^{-1}$. The overcorrelated zone therefore vanishes. 
However, the matching condition must be strictly satisfied; even slight nonuniformity in the weight vector $\mathbf{w}$ leads to a substantial degradation in the EQFI, which will be illustrated in the \textit{Example} section. 
The cases $r_i=0$ or $s_i=0$ are given in Sec.~III of \cite{supp}.

These regimes are summarized in the $(\mu_i,s_i/r_i)$ phase plane in \cref{fig:phasemap}(b). Maximizing \cref{eq: phase-map} over the admissible correlation strength yields the following node--subnetwork limit and its equality condition.

\begin{corollary}[Node--Subnetwork Limit]\label{col: maximum}\vspace{-2mm}
For fixed $\mathbf{Q}_{\not{i}}$ and correlation direction $\mathbf{v}_i$, the maximum global EQFI is
\begin{equation}\label{eq: subnetwork-limit}
    \E_{w,\max}=\begin{cases}
    A_{\not{i}}^{-1}, & |s_i|\leq r_i,\\[1mm]
    A_{\not{i}}^{-1}(1+s_i^2-r_i^2)^{-1},& |s_i|>r_i.
    \end{cases}
\end{equation}
The optimum occurs at $\mu_i=\mu_i^*$, with the matched case reached as a singular limit.
\vspace{-2mm}
\end{corollary}

This result also provides a geometric proof of \cref{thm: maximum}. For any participating node, $A_{\not{i}}^{-1}=s_i^2\kappa_i$. If $|s_i|\leq r_i\leq1$, then $\E_{w,\max}^{(i)}=s_i^2\kappa_i\leq\kappa_i$. If $|s_i|>r_i$, then $\E_{w,\max}^{(i)}\leq(A_{\not{i}}s_i^2)^{-1}=\kappa_i$. Applying this inequality to every node gives $\E_w\leq\min_i\kappa_i$. When $|s_i|\leq r_i$, the matching condition $\mu_i=s_i/r_i$ cancels the node-dependent mismatch term in \cref{eq: phase-map}, leaving the global precision limited solely by the remaining subnetwork. 

\textit{Optimal QFIM Design with Fixed Local Capacities---}
\cref{col: maximum} reduces the global optimization to a sequence of local matching conditions, which we impose recursively to construct QFIMs that attain the bottleneck bound in \cref{thm: maximum}.

For fixed local capacities and target weights, the bottleneck-saturating QFIM is generally nonunique. For convenience, we choose the maximally aligned case, $r_k=1$, at each recursive step.

We sort the participating nodes in decreasing order of their weighted capacities, $\kappa_1\geq\kappa_2\geq\cdots\geq\kappa_N$, and define $\mathbf{Q}^{(k)}$ and $\mathbf{w}^{(k)}=(w_k,\ldots,w_N)^T$, respectively, as the QFIM and weight vector of the subnetwork containing nodes $k,\ldots,N$. 
We then implement the node--subnetwork matching condition recursively, adding one node at a time while preserving the bottleneck-limited global EQFI.

We initialize the construction with $\mathbf{Q}^{(N)}=q_N$, for which $\E_w^{(N)}=q_N/w_N^2=\kappa_N$. Suppose that $\mathbf{Q}^{(k+1)}$ has been constructed with $\E_w^{(k+1)}=\kappa_N$. We add node $k$ through
\begin{equation}
    \mathbf{Q}^{(k)} = \begin{pmatrix}
        q_k & \mathbf{c}^{(k)T}\\
        \mathbf{c}^{(k)} & \mathbf{Q}^{(k+1)}
    \end{pmatrix}.
\end{equation}
The maximally aligned direction is $\mathbf{v}^{(k+1)}={\mathbf{w}^{(k+1)}}/{\|\mathbf{w}^{(k+1)}\|}$, 
while the signed coupling is
\begin{equation}\label{eq: recursive_correlation_main}
    \mathbf{c}^{(k)}=\kappa_Nw_k\mathbf{w}^{(k+1)}.
\end{equation}
Because $\E_w^{(k+1)}=\kappa_N$, the corresponding phase-map coordinates satisfy
$r_k=1$ and $\mu_k=s_k=w_k\sqrt{\kappa_N/q_k}
=\sgn(w_k)\sqrt{\kappa_N/\kappa_k}$.
Since $\kappa_k\geq\kappa_N$, one has $|\mu_k|\leq1$. The matching condition $\mu_k=s_k$ therefore places each added node at the subnetwork ceiling, preserving $\E_w^{(k)}=\E_w^{(k+1)}=\kappa_N$. 

Iterating from $k=N-1$ to $k=1$ yields the canonical QFIM
$\mathbf{Q}^{(1)}=\kappa_N\mathbf{w}\mathbf{w}^T+\operatorname{diag}\!\left(q_i-\kappa_Nw_i^2\right)$, which has $\mathbf{Q}^{(1)}_{ii}=q_i$ and saturates $\E_w=\kappa_N=\min_i\kappa_i$. 

More generally, each level permits any geometrically realizable $r_k\in[\sqrt{\kappa_N/\kappa_k},1]$, together with $\mu_k=s_k/r_k$, generating a nonunique family of optimal QFIMs; see Sec.~IV of \cite{supp}.

\begin{figure}[t]
\centering
\includegraphics[width=0.9\linewidth]{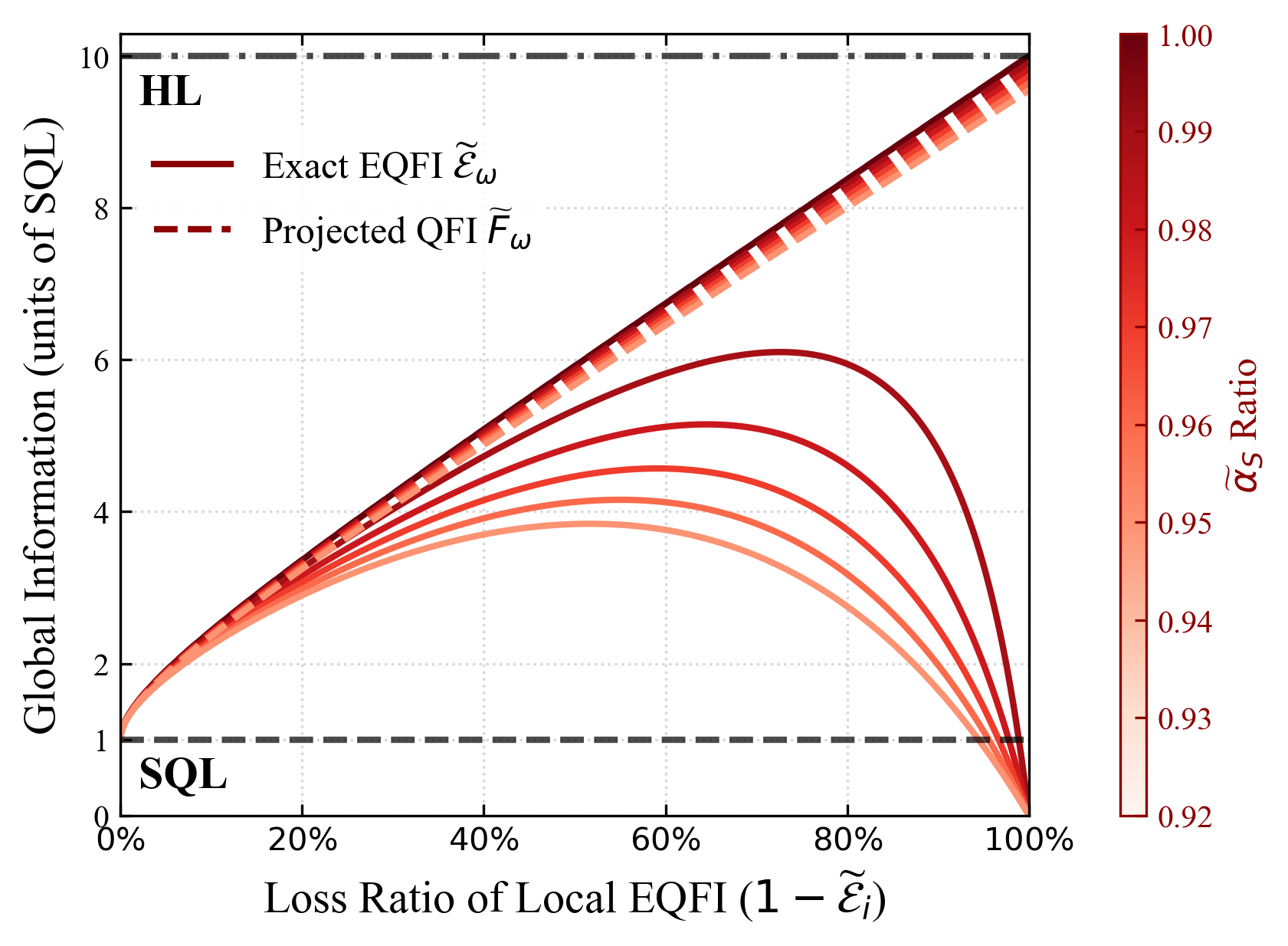}
\caption{
    \textit{EQFI and projected QFIs in a permutation-invariant network.}
    For $N=10$, the normalized EQFI $\widetilde{\E}_w$ (solid curves) and normalized projected QFI $\widetilde{F}_w$ (dashed curves) are plotted against the local-information loss $1-\widetilde{\E}_i$. Colors indicate the weight-uniformity parameter $\widetilde{\alpha}_S$, which accounts for the uniformity of $\mathbf{w}$.
    Only in the fully uniform and matched case ($\widetilde{\alpha}_S = 1$), $\widetilde{\mathcal{E}}_w$ and $\widetilde{F}_w$ coincide and exhibit monotonic increase.  
    For other unmatched scenarios with $\widetilde{\alpha}_S < 1$, $\widetilde{F}_w$ remains monotonic and exhibits only weak dependence on $\widetilde{\alpha}_S$, becoming invalid in nuisance-aware QSN. By contrast, $\widetilde{\mathcal{E}}_w$ attains a distinct maximum and declines rapidly as $\widetilde{\alpha}_S$ deviates from unity, thereby still serving as a faithful quantification.}
\vspace{-3mm}
\label{fig:permutation_invariant}
\end{figure}

\textit{Example: Permutation-Invariant Network---}
We illustrate the phase map using a permutation-invariant QFIM
\begin{equation}
    \mathbf{Q}=q\bigl[(1-\gamma)\mathbf{I}_N+\gamma\mathbf{J}_N\bigr],
    \quad (0\leq\gamma<1),
\end{equation}
where every node has diagonal QFI $q$ and every pair has correlation $q\gamma$. The singular limit $\gamma\rightarrow1$ gives the rank-one QFIM associated with an ideal GHZ-type probe.

For normalized nonnegative weights satisfying $\sum_iw_i=1$, define $S=\sum_iw_i^2$ and $\widetilde{\alpha}_S=\frac{1-S}{S(N-1)}$. 
Here, $\widetilde{\alpha}_S=0$ represents a fully localized target and $\widetilde{\alpha}_S=1$ represents a uniform global average, i.e., the matched case. The fractional local-information loss is $1-\widetilde{\E}_i=\frac{(N-1)\gamma^2}{1+(N-2)\gamma}$. 

Taking the uncorrelated value $\E_{w,\mathrm{SQL}}=q/S$ as the reference, we compare the nuisance-aware EQFI with the directional projected QFI $F_w=\mathbf{w}^T\mathbf{Q}\mathbf{w}/(\mathbf{w}^T\mathbf{w})^2$. Their normalized forms are
\begin{equation}\label{eq: invariant_enhancement}
    \begin{aligned}
        \widetilde{F}_w &\equiv \frac{F_w}{\E_{w,\mathrm{SQL}}} = 1+(N-1)\gamma\widetilde{\alpha}_S,\\
        \widetilde{\E}_w &\equiv \frac{\E_w}{\E_{w,\mathrm{SQL}}}= 1+(N-1)\gamma \frac{\widetilde{\alpha}_S-\gamma}{1+\gamma R},
    \end{aligned}
\end{equation}
where $R=(N-1)(1-\widetilde{\alpha}_S)-1$.

Only for a uniform target, i.e., $\widetilde{\alpha}_S=1$, the target direction coincides with the symmetric QFIM eigenmode and $F_w=\E_w$. For any nonuniform target, $\widetilde{\alpha}_S<1$, one instead finds $F_w>\E_w$ for $\gamma>0$. 
Thus, $F_w$, here, fails to faithfully reflect the attainable global accuracy.

As shown in \cref{fig:permutation_invariant}, for $0<\widetilde{\alpha}_S<1$, $\widetilde{F}_w$ increases monotonically with $\gamma$, whereas $\widetilde{\E}_w$ reaches a finite optimum before entering the overcorrelated regime. The nuisance-aware enhancement is therefore highly sensitive to nonuniformity in target weight $\mathbf{w}$. For example, at $N=10$, decreasing $\widetilde{\alpha}_S$ from $1$ to $0.99$ reduces the optimized EQFI by approximately $40\%$ relative to the uniform-target optimum. More generally, retaining HL scaling requires $1-\widetilde{\alpha}_S=O(1/N)$; see Sec.~V of \cite{supp}.

\textit{Local Privacy in QSN---}
The phase map further reveals that, at its matched-singular boundary, the global target may remain estimable even when every local parameter is unidentifiable. Motivated by this observation, we define \textit{intrinsic local privacy} by $\E_i=0$ for all $i$, while the target $\theta_w=\mathbf{w}^T\bm{\theta}$ retains $\E_w>0$. This is a QFIM-level non-identifiability condition rather than, by itself, a cryptographic security guarantee. 

\begin{theorem}[Intrinsic Local Privacy]
\label{thm:privacy_security}
Let $\mathbf{Q}\geq0$ and $\mathbf{w}\neq0$. Intrinsic local privacy holds if and only if
\begin{equation}\label{eq: local_privacy_range}
    \mathbf{e}_i\notin\operatorname{range}(\mathbf{Q}) \quad\forall i, \qquad \mathbf{w}\in\operatorname{range}(\mathbf{Q}).
\end{equation}
Within the node--subnetwork phase map, this corresponds to the
simultaneous matched-singular boundary
\begin{equation}\label{eq: privacy_phase_condition}
    |\mu_i|=1, \qquad s_i=\mu_i r_i, \qquad\forall i,
\end{equation}
understood through a nonsingular limiting sequence.
\end{theorem}
\begin{proof}
    For a singular QFIM, a scalar target $\theta_{\mathbf{a}}=\mathbf{a}^T\bm{\theta}$ is locally estimable if and only if $\mathbf{a}\in\operatorname{range}(\mathbf{Q})$, in which case $\E_{\mathbf{a}} =[\mathbf{a}^T\mathbf{Q}^{+}\mathbf{a}]^{-1}$; otherwise, $\E_{\mathbf{a}}=0$. Applying this criterion to $\mathbf{a}=\mathbf{e}_i$ and $\mathbf{a}=\mathbf{w}$ proves \cref{eq: local_privacy_range}. Moreover, $\E_i=q_i(1-\mu_i^2)$ shows that complete local-information loss requires $|\mu_i|=1$. At this singular boundary, the global EQFI in \cref{eq: phase-map} retains a finite nonzero value only if the mismatch term simultaneously vanishes, $s_i-\mu_i r_i=0$; otherwise, $\E_w\rightarrow0$. This proves \cref{eq: privacy_phase_condition}.
\end{proof}

Intrinsic local privacy therefore lies at the matched-singular boundary where all local information vanishes without an accompanying collapse of the global information. 

For comparison, previous studies introduced the stronger notion of \textit{functional privacy}, which requires every nonzero direction orthogonal to $\mathbf{w}$ to be inaccessible, $\E_v=0$ for all nonzero $\mathbf{v}\perp\mathbf{w}$, while the target remains estimable, $\E_w>0$~\cite{shettell2022private,hassani2025privacy,namkung2026universal,farokhi2026precision}. These requirements force $\operatorname{range}(\mathbf{Q})=\operatorname{span}\{\mathbf{w}\}$ and hence $\mathbf{Q}\propto\mathbf{w}\mathbf{w}^T$~\cite{hassani2025privacy}. For any nonlocalized target satisfying $\mathbf{w}\not\parallel\mathbf{e}_i$ for all $i$, this rank-one form also satisfies \cref{eq: local_privacy_range}; thus, functional privacy implies intrinsic local privacy. Further connections are discussed in Sec.~VI of \cite{supp}. %For fixed diagonal capacities, functional privacy is feasible only if $q_i=\kappa w_i^2$ for every node, so all nodes with nonzero target weights have the same weighted capacity, $q_i/w_i^2=\kappa$. The QFIM is then $\mathbf{Q}=\kappa\mathbf{w}\mathbf{w}^T$, which yields $\E_w=\kappa$ and saturates the bottleneck bound.

By contrast, intrinsic local privacy permits QFIMs of any rank below $N$. To demonstrate this explicitly, consider $N\geq4$ qubit sensors with $H_i=Z_i/2$ and $Q_{ii}=1$ for all $i$, together with the uniform target $\mathbf w=\mathbf1_N/N$. For any prescribed rank $r\in\{1,\ldots,N-1\}$, define the $r$ sign vectors:
\begin{equation}
\begin{aligned}
    \mathbf v^{(0)}=&\mathbf1_N,\\
    \mathbf v^{(j)}=&\mathbf1_N-2(\mathbf e_j+\mathbf e_N),
\end{aligned}
\end{equation}
for $j=1,...,r-1$. We prepare the probe
\begin{equation}\label{eq: privacy-construction}
    |\psi\rangle=\frac{1}{\sqrt{2r}}\sum_{j=0}^{r-1}\left(|\mathbf v^{(j)}\rangle_Z+|-\mathbf v^{(j)}\rangle_Z\right),
\end{equation}
where $|\mathbf v\rangle_Z$ denotes the joint computational-basis state satisfying $Z_i|\mathbf v\rangle_Z=v_i|\mathbf v\rangle_Z$. The symmetric support on $\pm\mathbf v^{(j)}$ gives $\langle Z_i\rangle=0$; hence the pure-state QFIM is therefore $\mathbf Q=\frac{1}{r}\sum_{j=0}^{r-1}\mathbf v^{(j)}\mathbf v^{(j)T}$, with $\operatorname{rank}(\mathbf Q)=r$. 
The span of these linearly independent vectors contains the target direction $\mathbf w=\mathbf v^{(0)}/N$ but contains no individual direction $\mathbf e_i$. Hence \cref{thm:privacy_security} gives $\E_i=0$ for every node, while $\E_w=N^2/r$. Thus, $r$ controls the number of estimable collective modes without compromising intrinsic local privacy. For $r=1$, \cref{eq: privacy-construction} reduces to a GHZ state. Further details are provided in Sec.~VII of \cite{supp}.

\textit{Discussion---}
Weighted DQS is intrinsically a function-estimation problem in the presence of nuisance parameters, when the local parameters are independently unknown. We derived an exact local--global phase map that determines how node--subnetwork correlations redistribute metrological information. The map identifies the optimal and break-even correlation strengths and reveals an overcorrelated regime in which further local-information loss also degrades global precision. For fixed diagonal sensing capacities, the nodewise constraints produce a tight weakest-capacity bottleneck, which is saturated by our recursive family of QFIM designs.

The same phase map identifies intrinsic local non-identifiability with the simultaneous matched-singular boundary. Functional privacy further restricts the QFIM to the target direction and, under fixed diagonals, requires equal weighted capacities. Intrinsic local privacy is less restrictive: our qubit construction realizes it at every QFIM rank from $1$ to $N-1$.

These results establish a correlation-geometric framework connecting resource-constrained precision, optimal network design, and intrinsic privacy. Because the analysis is formulated at the QFIM level, implementations on discrete-variable, continuous-variable, or hybrid platforms must additionally satisfy physical-realizability and measurement-compatibility conditions.

\textit{Acknowledgments---} This work was supported by the National Natural Science Foundation of China (Grant Nos.~12350006 and 92576202), the Quantum Science and Technology--National Science and Technology Major Project (Grant No.~2021ZD0301200), and the USTC Research Funds of the Double First-Class Initiative (Grant No.~YD2030002026).

\bibliography{references}

\end{document}

\title{Supplementary Material}

\maketitle

\tableofcontents

\section{Foundations of Distributed Quantum Metrology}

\subsection{General Framework of Distributed Quantum Sensing}

Distributed quantum sensing (DQS) extends conventional quantum metrology to networks of spatially separated sensor nodes~\cite{zhang2021distributed,ge2018distributed,proctor2018multiparameter,eldredge2018optimal}. We consider an $N$-node quantum sensor network in which node $i$ probes a local parameter $\theta_i$. The complete set of encoded parameters is collected into the vector
\begin{equation}
\boldsymbol{\theta}=(\theta_1,\theta_2,\ldots,\theta_N)^T\in \mathbb{R}^N.
\end{equation}

The network is initialized in a generally multipartite probe state $\rho_0$, which may contain correlations or entanglement across different nodes. We focus on local unitary parameter encoding, under which the output state is
\begin{equation}
\rho_{\boldsymbol{\theta}}=U(\boldsymbol{\theta})\rho_0U^\dagger(\boldsymbol{\theta}),
\end{equation}
with
\begin{equation}
U(\boldsymbol{\theta})=\exp\left(i\sum_{i=1}^{N}H_i\theta_i\right).
\end{equation}
Here, $H_i$ is a Hermitian generator acting nontrivially only on the Hilbert space of node $i$.

Because generators associated with distinct nodes act on disjoint tensor factors, they mutually commute,
\begin{equation}
[H_i,H_j]=0,
\qquad i\neq j.
\end{equation}
Consequently, the global encoding decomposes into parallel local evolutions,
\begin{equation}
U(\boldsymbol{\theta})
=
\bigotimes_{i=1}^{N}
\exp\left(iH_i\theta_i\right).
\end{equation}
The encoding dynamics are therefore local, whereas the probe state and the resulting measurement statistics may remain globally correlated through $\rho_0$.

The objective of DQS is often not to estimate every local parameter separately, but rather to infer a prescribed global property of the field. In this work, the target parameter is a weighted linear combination
\begin{equation}
\theta_w=\mathbf{w}^T\boldsymbol{\theta}=\sum_{i=1}^{N}w_i\theta_i,
\end{equation}
where
\begin{equation}
\mathbf{w}=(w_1,w_2,\ldots,w_N)^T\in\mathbb{R}^N
\end{equation}
is a known weight vector specifying the sensing task. Uniform positive weights describe, for example, estimation of a spatially averaged field, whereas nonuniform or signed weights can represent differential, gradient, or mode-selective sensing. No normalization of $\mathbf{w}$ is assumed unless stated explicitly.

Although the parameter of interest $\theta_w$ is scalar, the encoded state $\rho_{\boldsymbol{\theta}}$ generally depends on the full $N$-dimensional parameter vector. One may complete $\mathbf{w}$ to a basis of parameter space and write
\begin{equation}\boldsymbol{\phi}=\left(\theta_w,\phi_2,\ldots,\phi_N\right)^T,
\end{equation}
where $\phi_2,\ldots,\phi_N$ are independent linear combinations of the local parameters. When these additional combinations are unknown, they act as nuisance parameters in the estimation of $\theta_w$. Therefore, weighted DQS is intrinsically a scalar estimation task embedded in a multiparameter statistical model. Correctly accounting for the resulting nuisance-parameter penalty requires the effective quantum Fisher information developed below.

\subsection{Quantum Fisher Information Matrix and Global Precision Bounds}

The precision of multiparameter quantum estimation is characterized by the quantum Fisher information matrix (QFIM). For any estimator vector $\bm{\hat{\theta}}=(\hat{\theta}_1,\hat{\theta}_2,\ldots,\hat{\theta}_N)^T$ that is locally unbiased at the true parameter value, the symmetric-logarithmic-derivative quantum Cramér--Rao bound gives
\begin{equation}\label{eq: QCRB}
\mathrm{Cov}[\bm{\hat{\theta}}]\ge\frac{1}{\nu}\mathbf{Q}^{-1},
\end{equation}
where $\nu$ is the number of independent repetitions and the matrix inequality is understood in the positive-semidefinite sense. Here,
\begin{equation}
\mathrm{Cov}[\bm{\hat{\theta}}]_{ij}=\mathbb{E}\left[(\hat{\theta}_i-\theta_i)(\hat{\theta}_i-\theta_i)\right].
\end{equation}
\cref{eq: QCRB} assumes that $\mathbf{Q}$ is nonsingular. The corresponding formulation for singular QFIMs is given below.

For the parameterized state $\rho_{\bm{\theta}}$, the QFIM is the real symmetric positive-semidefinite matrix with elements
\begin{equation}
Q_{ij}=\frac{1}{2}\operatorname{Tr}\big[\rho_{\bm{\theta}}\{L_i,L_j\}\big],
\end{equation}
where the symmetric logarithmic derivative $L_i$ is defined implicitly by
\begin{equation}
\frac{\partial\rho_{\bm{\theta}}}{\partial\theta_i}=\frac{1}{2}\left(L_i\rho_{\bm{\theta}}+\rho_{\bm{\theta}}L_i\right).
\end{equation}
For a pure state $|\psi_{\bm{\theta}}\rangle$ undergoing unitary encoding, the QFIM reduces to
\begin{equation}
Q_{ij}=4\operatorname{Re}\left[\langle H_iH_j\rangle-\langle H_i\rangle\langle H_j\rangle\right].
\end{equation}
Because the local generators considered here commute, the covariance is real and the expression becomes $Q_{ij}=4\mathrm{Cov}(H_i,H_j)$.

The diagonal element
\begin{equation}
q_i\equiv Q_{ii}=4\mathrm{Var}(H_i)
\end{equation}
is the QFI associated with $\theta_i$ when all other parameters are regarded as fixed. In this work, $q_i$ is treated as the local sensing capacity determined by the physical resources available at node $i$. The off-diagonal elements $Q_{ij}$ characterize correlations between the parameter responses at different nodes. Although these correlations may originate from multipartite entanglement, nonzero off-diagonal QFIM elements are not, by themselves, an entanglement criterion.

The matrix bound in \cref{eq: QCRB} is not jointly attainable for an arbitrary multiparameter quantum model because the optimal measurements for different parameters may be incompatible. For the pure-state unitary model considered here, commuting local generators imply the weak-commutativity condition
\begin{equation}
\operatorname{Im}\langle\partial_i\psi_{\bm{\theta}}|\partial_j\psi_{\bm{\theta}}\rangle=0,
\end{equation}
and the SLD bound can therefore be attained asymptotically under the usual regularity conditions~\cite{holevo2011probabilistic,liu2020quantum}. For mixed-state probes, additional compatibility conditions may be required.

We now consider estimation of the weighted global parameter
\begin{equation}
\theta_w=\mathbf{w}^T\bm{\theta}.
\end{equation}
For the estimator $\hat{\theta}_w=\mathbf{w}^T\bm{\hat{\theta}}$, \cref{eq: QCRB} implies
\begin{equation}\label{eq: global_qcrb}
\Delta^2\hat{\theta}_w=\mathbf{w}^T\mathrm{Cov}[\bm{\hat{\theta}}]\mathbf{w}\ge\frac{1}{\nu}\mathbf{w}^T\mathbf{Q}^{-1}\mathbf{w}\equiv\frac{1}{\nu\mathcal{E}_w},
\end{equation}
where
\begin{equation}\label{eq: global_eqfi}
\mathcal{E}_w=\left(\mathbf{w}^T\mathbf{Q}^{-1}\mathbf{w}\right)^{-1}
\end{equation}
is the effective quantum Fisher information (EQFI) for the global target. The inverse QFIM in \cref{eq: global_eqfi} incorporates the precision penalty produced by the remaining $N-1$ unknown parameter combinations.

It is useful to distinguish the EQFI from the projected QFI
\begin{equation}\label{eq: projected_qfi}
F_w=\frac{\mathbf{w}^T\mathbf{Q}\mathbf{w}}{(\mathbf{w}^T\mathbf{w})^2},
\end{equation}
which frequently appears in distributed sensing~\cite{hyllus2012fisher,ge2018distributed,proctor2018multiparameter,eldredge2018optimal}. The quantity $F_w$ is the exact QFI of the one-parameter model. It assumes that every parameter combination orthogonal to $\mathbf{w}$ is known or physically constrained, so that the state varies only along the prescribed target direction. By contrast, $\mathcal{E}_w$ describes the genuinely multiparameter setting in which these orthogonal combinations remain unknown nuisance parameters.

For any positive-definite QFIM, the projected QFI upper-bounds the EQFI:
\begin{equation}\label{eq: inequality_chain}
F_w=\frac{\mathbf{w}^T\mathbf{Q}\mathbf{w}}{(\mathbf{w}^T\mathbf{w})^2}\ge\frac{1}{\mathbf{w}^T\mathbf{Q}^{-1}\mathbf{w}}=\mathcal{E}_w.
\end{equation}

\textit{Proof.} Applying the Cauchy--Schwarz inequality to the vectors $\mathbf{Q}^{1/2}\mathbf{w}$ and $\mathbf{Q}^{-1/2}\mathbf{w}$ gives
\begin{equation}
\left(\mathbf{w}^T\mathbf{Q}\mathbf{w}\right)\left(\mathbf{w}^T\mathbf{Q}^{-1}\mathbf{w}\right)\ge(\mathbf{w}^T\mathbf{w})^2.
\end{equation}
Rearranging immediately yields \cref{eq: inequality_chain}. Equality holds if and only if $\mathbf{Q}^{1/2}\mathbf{w}$ and $\mathbf{Q}^{-1/2}\mathbf{w}$ are linearly dependent, or equivalently, $\mathbf{Q}\mathbf{w}=\lambda\mathbf{w}$ for some $\lambda>0$. Hence, $F_w=\mathcal{E}_w$ precisely when the target vector $\mathbf{w}$ is an eigenvector of the QFIM. 

The inequality in \cref{eq: inequality_chain} does not imply that $F_w$ is intrinsically unphysical. Rather, $F_w$ and $\mathcal{E}_w$ answer different estimation questions. The projected quantity $F_w$ characterizes a controlled single-parameter model in which all nuisance directions are absent, whereas $\mathcal{E}_w$ gives the attainable information when the local parameters are independently unknown. Their discrepancy therefore quantifies the nuisance-parameter penalty generated by misalignment between the target direction and the information geometry of the network. This distinction underlies the local--global trade-off and the overcorrelated regime analyzed in the main text.

\subsection{Derivation of the Effective QFI and the Nuisance-Parameter Penalty}

Consider the estimation of a local parameter $\theta_i$. If $\theta_i$ is the only unknown parameter, while all remaining parameters are known or fixed, the single-parameter quantum Cramér--Rao bound gives
\begin{equation}
\Delta^2\hat{\theta}_i\ge\frac{1}{\nu q_i}, \qquad q_i=Q_{ii}=\operatorname{Tr}\left(\rho_{\bm{\theta}}L_i^2\right).
\end{equation}
Here, $q_i$ is the QFI associated with $\theta_i$ when all other parameters are treated as fixed.

When the complete parameter vector $\bm{\theta}=(\theta_1,\theta_2,\ldots,\theta_N)^T$ is unknown, the remaining parameters
\begin{equation}
\bm{\theta}_{\not{i}}=(\theta_1,\ldots,\theta_{i-1},\theta_{i+1},\ldots,\theta_N)^T
\end{equation}
act as nuisance parameters for the estimation of $\theta_i$. The multiparameter quantum Cramér--Rao bound,
\begin{equation}
\mathrm{Cov}[\bm{\hat{\theta}}]\ge\frac{1}{\nu}\mathbf{Q}^{-1},
\end{equation}
then implies
\begin{equation}\label{eq: local_variance_inverse}
\Delta^2\hat{\theta}_i\ge\frac{1}{\nu}(\mathbf{Q}^{-1})_{ii}.
\end{equation}

To evaluate the inverse matrix element, we partition the QFIM into the node--subnetwork form
\begin{equation}\label{eq: node_subnetwork_qfim}
\mathbf{Q}=\begin{pmatrix}q_i&\mathbf{c}_i^T\\\mathbf{c}_i&\mathbf{Q}_{\not{i}}\end{pmatrix},
\end{equation}
where $\mathbf{Q}_{\not{i}}$ is obtained by removing the $i$th row and column of $\mathbf{Q}$, and
\begin{equation}
\mathbf{c}_i=(Q_{i1},\ldots,Q_{i,i-1},Q_{i,i+1},\ldots,Q_{iN})^T
\end{equation}
collects the QFIM correlations between node $i$ and the remaining subnetwork.

Assuming that $\mathbf{Q}_{\not{i}}$ is invertible, the block-matrix inversion formula gives
\begin{equation}
(\mathbf{Q}^{-1})_{ii}=\left(q_i-\mathbf{c}_i^T\mathbf{Q}_{\not{i}}^{-1}\mathbf{c}_i\right)^{-1}.
\end{equation}
The effective quantum Fisher information for the local parameter $\theta_i$ is therefore
\begin{equation}\label{eq: local_eqfi}
\mathcal{E}_i\equiv[(\mathbf{Q}^{-1})_{ii}]^{-1}=q_i-\mathbf{c}_i^T\mathbf{Q}_{\not{i}}^{-1}\mathbf{c}_i.
\end{equation}
Accordingly, the variance bound becomes
\begin{equation}
\Delta^2\hat{\theta}_i\ge\frac{1}{\nu\mathcal{E}_i}.
\end{equation}

The quadratic form $\mathbf{c}_i^T\mathbf{Q}_{\not{i}}^{-1}\mathbf{c}_i$ quantifies the nuisance-parameter penalty induced by the statistical coupling between $\theta_i$ and the remaining unknown parameters. Since $\mathbf{Q}_{\not{i}}^{-1}$ is positive definite, and hence
\begin{equation}
    0\le\mathcal{E}_i\le q_i.
\end{equation}
Thus, $q_i$ quantifies the information available when the nuisance parameters are known, whereas $\mathcal{E}_i$ quantifies the information that remains accessible when they are unknown. The equality $\mathcal{E}_i=q_i$ holds if and only if $\mathbf{c}_i=\mathbf{0}$, namely, when $\theta_i$ is decoupled from all nuisance directions at the QFIM level.

It is useful to introduce the normalized local EQFI
\begin{equation}\label{eq: normalized_local_eqfi}
\widetilde{\mathcal{E}}_i\equiv\frac{\mathcal{E}_i}{q_i}=1-\frac{\mathbf{c}_i^T\mathbf{Q}_{\not{i}}^{-1}\mathbf{c}_i}{q_i},
\end{equation}
which represents the fraction of the local sensing capacity that remains operationally accessible in the full multiparameter model. Equivalently, the fractional local information loss is
\begin{equation}\label{eq: fractional_local_loss}
1-\widetilde{\mathcal{E}}_i=\frac{\mathbf{c}_i^T\mathbf{Q}_{\not{i}}^{-1}\mathbf{c}_i}{q_i}.
\end{equation}

\subsection{Singular QFIMs and the Moore--Penrose Pseudoinverse}

The preceding multiparameter analysis assumed that the quantum Fisher information matrix is positive definite, $\mathbf{Q}>0$, so that its inverse is well defined. In many quantum sensor networks, however, symmetries or perfect correlations can render the QFIM singular, $\det(\mathbf{Q})=0$~\cite{liu2020quantum,ye2022quantum,candeloro2024dimension}. Singular QFIMs are particularly relevant to the intrinsic privacy conditions discussed in the main text.

Geometrically, a singular QFIM indicates that the encoded state is locally insensitive to one or more directions in parameter space. Specifically, if a vector $\mathbf{u}$ satisfies
\begin{equation}
\mathbf{Q}\mathbf{u}=\mathbf{0},
\end{equation}
then the parameter combination associated with $\mathbf{u}$ cannot be locally distinguished at the operating point. Such a direction represents a redundancy of the statistical model or an unidentifiable combination of the physical parameters.

As an illustrative example, consider a two-node network initialized in the Bell state
\begin{equation}
|\psi_0\rangle=\frac{1}{\sqrt{2}}\left(|00\rangle+|11\rangle\right),
\end{equation}
with local generators $H_1=\sigma_z\otimes I$ and $H_2=I\otimes\sigma_z$. Under parallel unitary encoding, the state becomes
\begin{equation}
|\psi_{\bm{\theta}}\rangle=\exp\left[i(\theta_1H_1+\theta_2H_2)\right]|\psi_0\rangle=\frac{1}{\sqrt{2}}\left[e^{i(\theta_1+\theta_2)}|00\rangle+e^{-i(\theta_1+\theta_2)}|11\rangle\right].
\end{equation}
The encoded state depends only on the sum $\theta_+=\theta_1+\theta_2$ and is invariant under changes along the difference direction $\theta_-=\theta_1-\theta_2$. The corresponding QFIM is
\begin{equation}
\mathbf{Q}=4\begin{pmatrix}1&1\\ 1&1\end{pmatrix}.
\end{equation}
Its nonzero eigenvalue is $8$, associated with the estimable direction $(1,1)^T$, while its zero eigenvalue is associated with the null direction $(1,-1)^T$. The network can therefore resolve the sum of the two local parameters but cannot independently identify their difference.

When $\mathbf{Q}$ is singular, the ordinary inverse $\mathbf{Q}^{-1}$ does not exist. The estimation problem must instead be restricted to the support of the QFIM and described using the Moore--Penrose pseudoinverse.

Because $\mathbf{Q}$ is real, symmetric, and positive semidefinite, the full parameter space admits the orthogonal decomposition
\begin{equation}
\mathbb{R}^N=\operatorname{range}(\mathbf{Q})\oplus\operatorname{ker}(\mathbf{Q}).
\end{equation}
For the spectral decomposition $\mathbf{Q}=\sum_{k=1}^N\lambda_k|\nu_k\rangle\langle\nu_k|$, these subspaces are
\begin{equation}
\begin{aligned}
    \operatorname{range}(\mathbf{Q})=&\operatorname{span}\{|\nu_k\rangle:\lambda_k>0\},\\
    \operatorname{ker}(\mathbf{Q})=&\operatorname{span}\{|\nu_k\rangle:\lambda_k=0\}.
\end{aligned}
\end{equation}
The range of $\mathbf{Q}$ is the estimable subspace, whereas its kernel contains the locally unidentifiable parameter directions.

The Moore--Penrose pseudoinverse is obtained by inverting only the nonzero eigenvalues:
\begin{equation}\label{eq: qfim_pseudoinverse}
\mathbf{Q}^+=\sum_{\lambda_k>0}\lambda_k^{-1}|\nu_k\rangle\langle\nu_k|.
\end{equation}
It acts as the ordinary inverse on $\operatorname{range}(\mathbf{Q})$ and vanishes on $\operatorname{ker}(\mathbf{Q})$.

More generally, consider a scalar target parameter $\theta_{\mathbf{a}}=\mathbf{a}^T\bm{\theta}$,  specified by a vector $\mathbf{a}\in\mathbb{R}^N$. This target is locally estimable if and only if
\begin{equation}\label{eq: estimability_condition}
\mathbf{a}\in\operatorname{range}(\mathbf{Q}).
\end{equation}
Because $\mathbf{Q}$ is symmetric, this condition is equivalent to
\begin{equation}
\mathbf{a}^T\mathbf{u}=0 \qquad \forall \mathbf{u}\in\operatorname{ker}(\mathbf{Q}).
\end{equation}
Thus, an estimable target cannot contain any component along an unidentifiable parameter direction.

When \cref{eq: estimability_condition} is satisfied, the generalized quantum Cramér--Rao bound is
\begin{equation}\label{eq: singular_scalar_qcrb}
\Delta^2\hat{\theta}_{\mathbf{a}}\ge\frac{1}{\nu}\mathbf{a}^T\mathbf{Q}^+\mathbf{a}.
\end{equation}
The corresponding effective quantum Fisher information is therefore
\begin{equation}\label{eq: singular_general_eqfi}
\mathcal{E}_{\mathbf{a}}=\left(\mathbf{a}^T\mathbf{Q}^+\mathbf{a}\right)^{-1}, \qquad \mathbf{a}\in\operatorname{range}(\mathbf{Q}).
\end{equation}
If $\mathbf{a}\notin\operatorname{range}(\mathbf{Q})$, no locally unbiased estimator with finite variance exists for the complete target parameter. Its effective information is consequently defined to vanish:
\begin{equation}\label{eq: general_eqfi_piecewise}
\mathcal{E}_{\mathbf{a}}=\begin{cases}\left(\mathbf{a}^T\mathbf{Q}^+\mathbf{a}\right)^{-1},&\mathbf{a}\in\operatorname{range}(\mathbf{Q}),\\ 0,&\mathbf{a}\notin\operatorname{range}(\mathbf{Q}).\end{cases}
\end{equation}

For the weighted global target $\theta_w=\mathbf{w}^T\bm{\theta}$, \cref{eq: general_eqfi_piecewise} gives
\begin{equation}\label{eq: singular_global_eqfi}
\mathcal{E}_w=\begin{cases}\left(\mathbf{w}^T\mathbf{Q}^+\mathbf{w}\right)^{-1},&\mathbf{w}\in\operatorname{range}(\mathbf{Q}),\\ 0,&\mathbf{w}\notin\operatorname{range}(\mathbf{Q}).\end{cases}
\end{equation}
Throughout this work, whenever a finite global precision is considered, we assume that $\mathbf{w}\in\operatorname{range}(\mathbf{Q})$.

The same construction determines the local EQFI in singular cases. Estimation of the individual parameter $\theta_i$ corresponds to the target vector $\mathbf{e}_i$, where $\mathbf{e}_i$ is the $i$th Cartesian basis vector. Hence,
\begin{equation}\label{eq: singular_local_eqfi}
\mathcal{E}_i=\begin{cases}\left(\mathbf{e}_i^T\mathbf{Q}^+\mathbf{e}_i\right)^{-1}=[(\mathbf{Q}^+)_{ii}]^{-1},&\mathbf{e}_i\in\operatorname{range}(\mathbf{Q}),\\ 0,&\mathbf{e}_i\notin\operatorname{range}(\mathbf{Q}).\end{cases}
\end{equation}

This expression provides the singular extension of the local EQFI and directly yields the local-privacy criterion used in the main text: an individual parameter is metrologically inaccessible precisely when $\mathbf{e}_i\notin\operatorname{range}(\mathbf{Q})$, even though a global target may remain estimable when $\mathbf{w}\in\operatorname{range}(\mathbf{Q})$.

\section{Local Resource Constraints and Proof of Theorem I}

\subsection{Local Resources and Local Sensing Capacities}

The diagonal QFIM element $q_i\equiv Q_{ii}$ is the single-parameter QFI associated with $\theta_i$ when all other parameters are regarded as fixed. In this work, $q_i$ is treated as the local sensing capacity of node $i$. It is determined by the local generator, the physical resources assigned to the node, and the allowed probe-state preparation. Fixing the set ${q_i}$ therefore provides a general QFIM-level description of locally constrained quantum sensor networks.

For a pure probe undergoing local unitary encoding $U(\bm{\theta})=\bigotimes_{i=1}^N\exp(-i\theta_iH_i)$, the diagonal QFIM element is
\begin{equation}
q_i=4\Delta^2H_i=4\left(\langle H_i^2\rangle-\langle H_i\rangle^2\right).
\end{equation}
This relation connects the abstract capacity $q_i$ to the resources available in different sensing architectures.

\textit{Discrete-variable resources.---} Consider node $i$ containing $n_i$ two-level probes, with local generator $H_i=\frac{1}{2}\sum_{k=1}^{n_i}\sigma_z^{(k)}$. 
For product probes optimized for local phase estimation, the generator variance is $\Delta^2H_i=n_i/4$, giving the standard-quantum-limit capacity
\begin{equation}
q_i=n_i.
\end{equation}
More generally, $q_i\le n_i$ for pure product probes under this generator. If the node supports intra-node entanglement, a GHZ state maximizes the generator variance, $\Delta^2H_i=n_i^2/4$, and yields
\begin{equation}
q_i=n_i^2.
\end{equation}
Thus, for a fixed local particle number, the accessible capacity ranges from SQL scaling for product probes to Heisenberg scaling for optimally entangled probes~\cite{giovannetti2006quantum,giovannetti2011advances,pezze2018quantum}.

\textit{Continuous-variable resources.---} Consider an optical node in which a local phase is generated by the photon-number operator $H_i=\hat{n}_i=\hat{a}_i^\dagger\hat{a}_i$. 
Assuming access to an appropriate phase reference, a coherent state $|\alpha_i\rangle$ with mean photon number $E_i=\langle\hat{n}_i\rangle=|\alpha_i|^2$ has $\Delta^2\hat{n}_i=E_i$ and therefore
\begin{equation}
q_i=4E_i.
\end{equation}
For a single-mode squeezed-vacuum state with squeezing parameter $r_i$, $E_i=\sinh^2r_i$ and $ \Delta^2\hat{n}_i=2E_i(E_i+1)$, and hence
\begin{equation}
    q_i=8E_i(E_i+1)=8\sinh^2r_i\cosh^2r_i.
\end{equation}
In the strongly squeezed regime, $r_i\gg1$, this becomes
\begin{equation}
    q_i\simeq\frac{1}{2}e^{4r_i}\simeq8E_i^2.
\end{equation}
These examples illustrate the linear energy scaling of coherent probes and the quadratic energy scaling attainable within the squeezed-vacuum family~\cite{zhuang2018distributed}.

Our analysis avoids this ambiguity by taking $q_i$ itself as the fixed operational capacity of node $i$, irrespective of the microscopic resource responsible for that capacity.

In both discrete- and continuous-variable systems, $q_i$ incorporates the amount of local physical resources together with the local state-preparation capability. Expressing the global precision bound solely in terms of ${q_i}$ separates two optimization levels: \textit{local probe engineering determines the diagonal capacities, whereas inter-node correlations determine how these capacities are redistributed across local and global parameter directions.}

\subsection{Proof of Theorem I}

Let $\mathbf{e}_i$ denote the $i$th Cartesian basis vector, so that $w_i=\mathbf{e}_i^T\mathbf{w}$ and $q_i=\mathbf{e}_i^T\mathbf{Q}\mathbf{e}_i$. We first assume that $\mathbf{Q}>0$. Applying the Cauchy--Schwarz inequality gives
\begin{equation}
w_i^2=\left(\mathbf{e}_i^T\mathbf{Q}^{1/2}\mathbf{Q}^{-1/2}\mathbf{w}\right)^2\le\left(\mathbf{e}_i^T\mathbf{Q}\mathbf{e}_i\right)\left(\mathbf{w}^T\mathbf{Q}^{-1}\mathbf{w}\right)=\frac{q_i}{\mathcal{E}_w}.
\end{equation}
Therefore, for every node with $w_i\neq0$,
\begin{equation}\label{eq: individual_bottleneck}
\mathcal{E}_w\le\frac{q_i}{w_i^2}.
\end{equation}

The same result holds when $\mathbf{Q}$ is singular, provided that the global target is estimable, $\mathbf{w}\in\operatorname{range}(\mathbf{Q})$. In this case,
\begin{equation}
\mathbf{w}=\mathbf{Q}^{1/2}(\mathbf{Q}^{+})^{1/2}\mathbf{w},
\end{equation}
and hence
\begin{equation}
w_i^2=\left(\mathbf{e}_i^T\mathbf{Q}^{1/2}(\mathbf{Q}^{+})^{1/2}\mathbf{w}\right)^2\le\left(\mathbf{e}_i^T\mathbf{Q}\mathbf{e}_i\right)\left(\mathbf{w}^T\mathbf{Q}^+\mathbf{w}\right)=\frac{q_i}{\mathcal{E}_w}.
\end{equation}
Thus, \cref{eq: individual_bottleneck} remains valid for both nonsingular and singular QFIMs.

Define the intrinsic weighted capacity
\begin{equation}
\kappa_i=\begin{cases}q_i/w_i^2,&w_i\neq0,\\
+\infty,&w_i=0.\end{cases}
\end{equation}
Since \cref{eq: individual_bottleneck} holds for every node, the global EQFI satisfies
\begin{equation}\label{eq: bottleneck_bound_sm}
\mathcal{E}_w\le\kappa_{\min}, \qquad \kappa_{\min}\equiv\min_i\kappa_i.
\end{equation}

The physical interpretation is direct. Node $i$ has local sensing capacity $q_i$, while the global task requires that its parameter enter with weight $w_i$. The ratio $\kappa_i=q_i/w_i^2$ therefore measures the maximum global information that node $i$ can support under the prescribed task. The smallest $\kappa_i$ determines the network precision, independently of how strongly the remaining nodes are correlated. This is the quantum-sensing analogue of a barrel effect: the global capacity is fixed by the weakest required component.

The theorem also specifies when Heisenberg scaling with the network size is attainable. Suppose the local capacities are uniformly bounded as
\begin{equation}
0<q_{\min}\le q_i\le q_{\max}<\infty,
\end{equation}
independently of $N$, and the target weights are normalized and delocalized such that $\sum_{i=1}^N |w_i|=1$ and $ \max_i|w_i|=\Theta(1/N)$. 
Then $\kappa_{\min}=\Theta(N^2)$, and the tightness of the bottleneck bound gives
\begin{equation}
\mathcal{E}_{w,\max}=\Theta(N^2).
\end{equation}
For a fixed number of repetitions, the corresponding variance scales as
\begin{equation}
\Delta^2\hat{\theta}_w\ge\frac{1}{\nu\mathcal{E}_{w,\max}}=\Theta(N^{-2}),
\end{equation}
which is the Heisenberg scaling associated with a delocalized weighted sensing task.

\section{Derivation and Analysis of the Node--Subnetwork Phase Map}\label{sec: phase-map}

We derive the exact relationship between the local EQFI $\mathcal{E}_i$ of an arbitrary node $i$ and the global EQFI $\mathcal{E}_w$ associated with the weighted parameter $\theta_w=\mathbf{w}^T\bm{\theta}$. Consider the node--subnetwork decompositions
\begin{equation}
\mathbf{Q}=\begin{pmatrix}q_i&\mathbf{c}_i^T\\ \mathbf{c}_i&\mathbf{Q}_{\not{i}}\end{pmatrix},\qquad \mathbf{w}=\begin{pmatrix}w_i\\ \mathbf{w}_{\not{i}}\end{pmatrix},
\end{equation}
where $\mathbf{Q}_{\not{i}}\in\mathbb{R}^{(N-1)\times(N-1)}$ and $\mathbf{w}_{\not{i}}\in\mathbb{R}^{N-1}$ describe the complementary subnetwork. We assume throughout the main derivation that $q_i>0$, $\mathbf{Q}_{\not{i}}>0$, and $\mathbf{w}_{\not{i}}\neq\mathbf{0}$. Singular QFIMs are treated using the estimability conditions and Moore--Penrose pseudoinverse introduced in the preceding section.

\subsection{Derivation of the Phase Map}

Choose a unit vector $\mathbf{v}_i$ along the node--subnetwork correlation axis and decompose the correlation vector as $\mathbf{c}_i=c_i\mathbf{v}_i$ with $c_i\in\mathbb{R}$ and $|\mathbf{v}_i|=1$, where the scalar amplitude $c_i$ carries the sign of the correlation. The orientation of $\mathbf{v}_i$ is chosen such that $D_{\not{i}}\equiv\mathbf{w}_{\not{i}}^T\mathbf{Q}_{\not{i}}^{-1}\mathbf{v}_i\ge0$. 
Thus, positive and negative values of $c_i$ correspond, respectively, to correlations aligned and misaligned with the complementary target direction. When $\mathbf{c}_i=\mathbf{0}$, the choice of $\mathbf{v}_i$ is immaterial and the results below are recovered by setting $c_i=0$.

The local EQFI follows from the Schur complement:
\begin{equation}\label{eq:local_eqfi_schur_phase}
\mathcal{E}_i=q_i-\mathbf{c}_i^T\mathbf{Q}_{\not{i}}^{-1}\mathbf{c}_i=q_i-c_i^2B_{\not{i}},
\end{equation}
where
\begin{equation}\label{eq:phase_map_ABD}
A_{\not{i}}\equiv\mathbf{w}_{\not{i}}^T\mathbf{Q}_{\not{i}}^{-1}\mathbf{w}_{\not{i}},\qquad B_{\not{i}}\equiv\mathbf{v}_i^T\mathbf{Q}_{\not{i}}^{-1}\mathbf{v}_i.
\end{equation}
Under the assumptions $q_i>0$, $\mathbf{Q}_{\not{i}}>0$, and $\mathbf{w}_{\not{i}}\neq\mathbf{0}$, both $A_{\not{i}}$ and $B_{\not{i}}$ are strictly positive. Moreover, $A_{\not{i}}$, $B_{\not{i}}$, and $D_{\not{i}}$ form the Gram matrix $\mathbf{G}$
\begin{equation}
\mathbf{G} = \begin{pmatrix}A_{\not{i}}&D_{\not{i}}\\ D_{\not{i}}&B_{\not{i}}\end{pmatrix}\succeq0,
\end{equation}
which implies $D_{\not{i}}^2\le A_{\not{i}}B_{\not{i}}$.

We introduce the dimensionless parameters
\begin{equation}\label{eq:phase_parameters_sm}
\mu_i\equiv \frac{c_i}{\sqrt{q_i/B_{\not i}}},\qquad r_i\equiv\frac{D_{\not{i}}}{\sqrt{A_{\not{i}}B_{\not{i}}}},\qquad s_i\equiv\frac{w_i}{\sqrt{q_iA_{\not{i}}}}.
\end{equation}

Here, $\mu_i\in [-1,1]$ denotes the normalized correlation strength between the node and the sub-network. $\mu=+1$ means the maximum aligned correlation, $\mu=-1$ means the maximum mis-aligned correlation, and $\mu=0$ means the decoupled case.

The Gram-matrix inequality gives $0\le r_i\le1$. Meanwhile, positive semidefiniteness of the full QFIM requires the Schur complement in \cref{eq:local_eqfi_schur_phase} to be nonnegative, and hence $|\mu_i|\le1$. The local phase map therefore becomes
\begin{equation}\label{eq:local_map}
\mathcal{E}_i(\mu_i)=q_i(1-\mu_i^2).
\end{equation}

We next determine the global EQFI. For $|\mu_i|<1$, the Schur complement is strictly positive, and block inversion gives
\begin{equation}\label{eq:global_inverse_block}
\mathcal{E}_w^{-1}=\mathbf{w}^T\mathbf{Q}^{-1}\mathbf{w}=A_{\not{i}}+\frac{\left(w_i-\mathbf{c}_i^T\mathbf{Q}_{\not{i}}^{-1}\mathbf{w}_{\not{i}}\right)^2}{\mathcal{E}_i}.
\end{equation}
The second term represents the inverse-precision penalty associated with incorporating node $i$. From the definitions in \cref{eq:phase_parameters_sm},
\begin{equation}
w_i=s_i\sqrt{q_iA_{\not{i}}},\qquad \mathbf{c}_i^T\mathbf{Q}_{\not{i}}^{-1}\mathbf{w}_{\not{i}}=c_iD_{\not{i}}=\mu_ir_i\sqrt{q_iA_{\not{i}}}.
\end{equation}
Substituting these identities and \cref{eq:local_map} into \cref{eq:global_inverse_block} yields
\begin{equation}
\mathcal{E}_w^{-1}=A_{\not{i}}\left[1+\frac{(s_i-r_i\mu_i)^2}{1-\mu_i^2}\right].
\end{equation}
The exact global phase map is consequently
\begin{equation}\label{eq:global_map}
\mathcal{E}_w(\mu_i)=\frac{1}{A_{\not{i}}}\left[1+\frac{(s_i-r_i\mu_i)^2}{1-\mu_i^2}\right]^{-1}.
\end{equation}
Together, \cref{eq:local_map} and \cref{eq:global_map} reproduce the node--subnetwork phase map stated in Theorem II in the main text. The singular endpoints $|\mu_i|=1$ are understood through the corresponding pseudoinverse estimability condition and are analyzed separately below.

\subsection{Sign Symmetry}

The phase map is invariant under the simultaneous transformation
\begin{equation}
s_i\rightarrow-s_i,\qquad \mu_i\rightarrow-\mu_i.
\end{equation}
Indeed, $\mathcal{E}_i$ depends only on $\mu_i^2$, whereas $\mathcal{E}_w$ depends on $s_i$ and $\mu_i$ through $(s_i-r_i\mu_i)^2$. The phase-map trajectory for $s_i<0$ is therefore identical as a set of points to that for $|s_i|$, but it is traversed with the aligned and misaligned correlation branches interchanged.

Accordingly, we take $s_i>0$ in the following analysis without loss of generality. For arbitrary signed weights, the aligned branch is defined by $\sgn(\mu_i)=\sgn(s_i)$, and the optimal signed correlation is restored by multiplying the expressions below by $\sgn(s_i)$.

\subsection{Sub-References, Trade-Off and overcorrelated Regimes}

For fixed $s_i$ and $r_i$, varying $\mu_i$ changes the strength and sign of the node--subnetwork correlation while leaving the task profile and correlation axis fixed. Differentiating \cref{eq:local_map,eq:global_map} gives
\begin{equation}\label{eq:phase_derivatives}
\frac{d\mathcal{E}_i}{d\mu_i}=-2q_i\mu_i,\qquad \frac{d\mathcal{E}_w}{d\mu_i}=-\frac{2(s_i-r_i\mu_i)(s_i\mu_i-r_i)}{A_{\not{i}}(1-\mu_i^2)^2}\left[1+\frac{(s_i-r_i\mu_i)^2}{1-\mu_i^2}\right]^{-2}.
\end{equation}
For $s_i,r_i>0$, the algebraic extreme points of $\mathcal{E}_w$ are
\begin{equation}
\mu_i=\frac{s_i}{r_i},\text{ or } \mu_i=\frac{r_i}{s_i}.
\end{equation}
Their product is unity, so for $s_i\neq r_i$ exactly one lies in the physical interval $(0,1)$. The unique physical maximum occurs at
\begin{equation}\label{eq:optimal_mu_sm}
\mu_i^*=\min\left\{\frac{s_i}{r_i},\frac{r_i}{s_i}\right\}.
\end{equation}
For a general signed weight,
\begin{equation}
\mu_i^*=\sgn(s_i)\min\left\{\left|\frac{s_i}{r_i}\right|,\left|\frac{r_i}{s_i}\right|\right\}.
\end{equation}

At the decoupled point $\mu_i=0$,
\begin{equation}\label{eq:decoupled_reference_sm}
\mathcal{E}_i=q_i,\qquad \mathcal{E}_{w,\mathrm{dec}}\equiv\mathcal{E}_w(0)=\frac{1}{A_{\not{i}}(1+s_i^2)}.
\end{equation}
The quantity $\mathcal{E}_{w,\mathrm{dec}}$ is a decoupled reference rather than a universal SQL, because the complementary subnetwork described by $\mathbf{Q}_{\not{i}}$ may itself contain nonclassical correlations.

For the misaligned branch $\mu_i<0$, both EQFIs increase as $\mu_i$ approaches zero:
\begin{equation}
\frac{d\mathcal{E}_i}{d\mu_i}>0,\qquad \frac{d\mathcal{E}_w}{d\mu_i}>0,\qquad -1<\mu_i<0.
\end{equation}
Moreover,
\begin{equation}
\mathcal{E}_w(\mu_i)<\mathcal{E}_{w,\mathrm{dec}},\qquad -1<\mu_i<0,
\end{equation}
so the entire misaligned branch lies in the sub-reference regime.

On the aligned branch, $0<\mu_i<1$, the local EQFI decreases monotonically. Before the optimal correlation is reached,
\begin{equation}
0<\mu_i<\mu_i^*:\qquad \frac{d\mathcal{E}_i}{d\mu_i}<0,\qquad \frac{d\mathcal{E}_w}{d\mu_i}>0.
\end{equation}
Thus,
\begin{equation}
\frac{d\mathcal{E}_w}{d\mathcal{E}_i}<0,
\end{equation}
which defines the local--global trade-off regime: increasing the global EQFI is accompanied by a loss of locally accessible information.

Beyond the optimum,
\begin{equation}
\mu_i^*<\mu_i<1:\qquad \frac{d\mathcal{E}_i}{d\mu_i}<0,\qquad \frac{d\mathcal{E}_w}{d\mu_i}<0,
\end{equation}
and consequently
\begin{equation}
\frac{d\mathcal{E}_w}{d\mathcal{E}_i}>0.
\end{equation}
Both local and global EQFIs now decrease as the correlation strength grows. This is the overcorrelated regime.

For $\mu_i\neq0$, the slope of the phase-map trajectory is
\begin{equation}
\frac{d\mathcal{E}_w}{d\mathcal{E}_i}=\frac{d\mathcal{E}_w/d\mu_i}{d\mathcal{E}_i/d\mu_i}.
\end{equation}
At $\mu_i=0$, $d\mathcal{E}_i/d\mu_i=0$, while $d\mathcal{E}_w/d\mu_i\neq0$ for generic $sr\neq0$. The trajectory therefore has a vertical tangent at the decoupled point rather than a stationary slope.

The global EQFI returns to the decoupled reference on the aligned branch at the break-even correlation
\begin{equation}\label{eq:break_even_sm}
\mu_{i}^{\mathrm{BE}}=\frac{2sr}{s_i^2+r_i^2} = \frac{2}{s_i/r_i+r_i/s_i}.
\end{equation}
For $s_i,r_i>0$ and $s_i\neq r_i$,
\begin{equation}
0<\mu_i^*\le\mu_{i}^{\mathrm{BE}}<1.
\end{equation}
Thus,
\begin{equation}
0<\mu_i<\mu_{i}^{\mathrm{BE}}\Rightarrow\mathcal{E}_w>\mathcal{E}_{w,\mathrm{dec}},\qquad \mu_{i}^{\mathrm{BE}}<\mu_i<1\Rightarrow\mathcal{E}_w<\mathcal{E}_{w,\mathrm{dec}}.
\end{equation}

\begin{figure}[t]
\centering
\includegraphics[width=\linewidth]{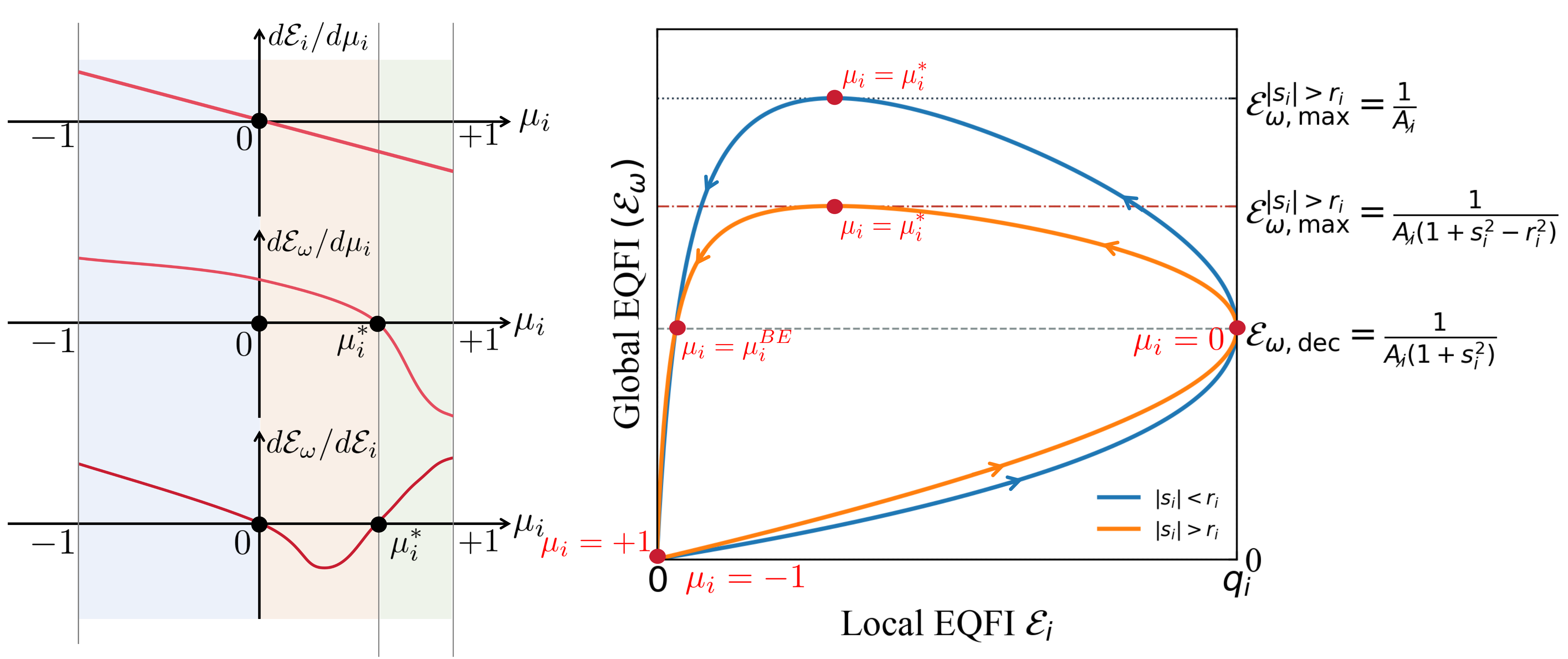}
\caption{\textit{Generic node--subnetwork phase map.} (a) Signs of $d\mathcal{E}_i/d\mu_i$, $d\mathcal{E}_w/d\mu_i$, and $d\mathcal{E}_w/d\mathcal{E}_i$ for $s_i>0$, $r_i>0$, and $s_i\neq r_i$. (b) Corresponding parametric trajectory $(\mathcal{E}_i(\mu_i),\mathcal{E}_w(\mu_i))$. The aligned branch first enters the local--global trade-off regime, reaches its maximum at $\mu_i^*$, and subsequently enters the overcorrelated regime.}
\label{fig: parameter-function-1}
\end{figure}

\subsection{Maximum Global EQFI}

The value of the maximum depends on the relative magnitudes of $s_i$ and $r_i$.

\textit{Subnetwork-ceiling branch, $0<s_i\le r_i$.---} The physical stationary point is
\begin{equation}
\mu_i^*=\frac{s_i}{r_i},
\end{equation}
at which the node-dependent contribution in \cref{eq:global_inverse_block} vanishes:
\begin{equation}
s_i-r_i\mu_i^*=0.
\end{equation}
The corresponding EQFIs are
\begin{equation}\label{eq:subnetwork_ceiling_sm}
\mathcal{E}_i(\mu_i^*)=q_i\left(1-\frac{s_i^2}{r_i^2}\right),\qquad \mathcal{E}_{w,\max}=\frac{1}{A_{\not{i}}}.
\end{equation}
The global precision reaches the absolute ceiling set by the complementary subnetwork.

\textit{Bottleneck-limited branch, $s_i>r_i$.---} The physical stationary point is
\begin{equation}
\mu_i^*=\frac{r_i}{s_i},
\end{equation}
and the maximum becomes
\begin{equation}\label{eq:bottleneck_branch_sm}
\mathcal{E}_i(\mu_i^*)=q_i\left(1-\frac{r_i^2}{s_i^2}\right),\qquad \mathcal{E}_{w,\max}=\frac{1}{A_{\not{i}}(1+s_i^2-r_i^2)}.
\end{equation}
Since $r_i\le1$,
\begin{equation}
\mathcal{E}_{w,\max}\le\frac{1}{A_{\not{i}}s_i^2}=\frac{q_i}{w_i^2},
\end{equation}
which recovers the node-wise bottleneck bound.

The boundary between the two branches is determined by
\begin{equation}
|s_i|\lessgtr r_i\quad\Longleftrightarrow\quad\frac{|w_i|}{\sqrt{q_i}}\lessgtr\frac{|D_{\not{i}}|}{\sqrt{B_{\not{i}}}}.
\end{equation}
The left-hand side is fixed by the weighted capacity of node $i$, whereas the right-hand side characterizes the alignment between the complementary weight profile and the node--subnetwork correlation axis.

For $s_i\neq r_i$, both EQFIs vanish at the aligned singular boundary:
\begin{equation}
\lim_{\mu_i\rightarrow1}\mathcal{E}_i(\mu_i)=0,\qquad \lim_{\mu_i\rightarrow1}\mathcal{E}_w(\mu_i)=0.
\end{equation}
The strong-correlation boundary therefore preserves finite global information only under the matched condition considered next.

\subsection{Matched Singular Boundary}

The matched case $|s_i|=r_i$ is the critical boundary at which the overcorrelated regime disappears. By the sign symmetry, it is sufficient to consider $s_i=r_i>0$. \cref{eq:global_map} then becomes
\begin{equation}
\mathcal{E}_w(\mu_i)=\frac{1}{A_{\not{i}}}\left[1+r_i^2\frac{1-\mu_i}{1+\mu_i}\right]^{-1}.
\end{equation}
It increases monotonically on the aligned branch and satisfies
\begin{equation}\label{eq:matched_limit_sm}
\lim_{\mu_i\rightarrow1}\mathcal{E}_w(\mu_i)=\frac{1}{A_{\not{i}}},\qquad \lim_{\mu_i\rightarrow1}\mathcal{E}_i(\mu_i)=0.
\end{equation}
Thus, the trade-off regime extends to the singular endpoint $\mu_i^*=1$, and the overcorrelated interval collapses.

At $\mu_i=1$, the Schur complement vanishes and the ordinary inverse of $\mathbf{Q}$ does not exist. We now verify \cref{eq:matched_limit_sm} directly using the pseudoinverse. At this point,
\begin{equation}
\mathbf{c}_i=\sqrt{\frac{q_i}{B_{\not{i}}}}\mathbf{v}_i,
\end{equation}
and the full QFIM has the null vector
\begin{equation}\label{eq:matched_null_vector}
\mathbf{u}_0=\begin{pmatrix}1\ -\mathbf{Q}_{\not{i}}^{-1}\mathbf{c}_i\end{pmatrix}=\begin{pmatrix}1\ -\sqrt{\frac{q_i}{B_{\not{i}}}}\,\mathbf{Q}_{\not{i}}^{-1}\mathbf{v}_i\end{pmatrix}.
\end{equation}
Indeed, $\mathbf{Q}\mathbf{u}_0=\mathbf{0}$ because the Schur complement is zero. Since $\mathbf{Q}_{\not{i}}>0$, the full QFIM has rank $N-1$ and a one-dimensional kernel spanned by $\mathbf{u}_0$.

The matching condition $s_i=r_i$ implies
\begin{equation}
w_i=\sqrt{\frac{q_i}{B_{\not{i}}}}D_{\not{i}}=\mathbf{c}_i^T\mathbf{Q}_{\not{i}}^{-1}\mathbf{w}_{\not{i}},
\end{equation}
and hence
\begin{equation}
\mathbf{w}^T\mathbf{u}_0=w_i-\mathbf{c}_i^T\mathbf{Q}_{\not{i}}^{-1}\mathbf{w}_{\not{i}}=0.
\end{equation}
Therefore, $\mathbf{w}\in\operatorname{range}(\mathbf{Q})$, and the global target remains estimable.

Now define
\begin{equation}
\mathbf{z}=\begin{pmatrix}0\\ \mathbf{Q}_{\not{i}}^{-1}\mathbf{w}_{\not{i}}\end{pmatrix}.
\end{equation}
The matching relation gives
\begin{equation}
\mathbf{Q}\mathbf{z}=\begin{pmatrix}\mathbf{c}_i^T\mathbf{Q}_{\not{i}}^{-1}\mathbf{w}_{\not{i}}\\ \mathbf{w}_{\not{i}}\end{pmatrix}=\mathbf{w}.
\end{equation}
Because $\mathbf{z}-\mathbf{Q}^+\mathbf{w}\in\operatorname{ker}(\mathbf{Q})$ and $\mathbf{w}\perp\operatorname{ker}(\mathbf{Q})$, it follows that
\begin{equation}
\mathbf{w}^T\mathbf{Q}^+\mathbf{w}=\mathbf{w}^T\mathbf{z}=\mathbf{w}_{\not{i}}^T\mathbf{Q}_{\not{i}}^{-1}\mathbf{w}_{\not{i}}=A_{\not{i}}.
\end{equation}
Therefore,
\begin{equation}
\mathcal{E}_w=\left(\mathbf{w}^T\mathbf{Q}^+\mathbf{w}\right)^{-1}=\frac{1}{A_{\not{i}}},
\end{equation}
in agreement with the limiting result.

By contrast,
\begin{equation}
\mathbf{e}_i^T\mathbf{u}_0=1\neq0,
\end{equation}
so $\mathbf{e}_i\notin\operatorname{range}(\mathbf{Q})$ and the individual parameter $\theta_i$ is not estimable:
\begin{equation}
\mathcal{E}_i=0.
\end{equation}
The matched boundary therefore realizes complete local inaccessibility for node $i$ while retaining finite global information. Complete local privacy of the entire network requires this estimability condition to hold for every node, as discussed in \cref{sec: privacy}.

\begin{figure}[t]
\centering
\includegraphics[width=0.7\linewidth]{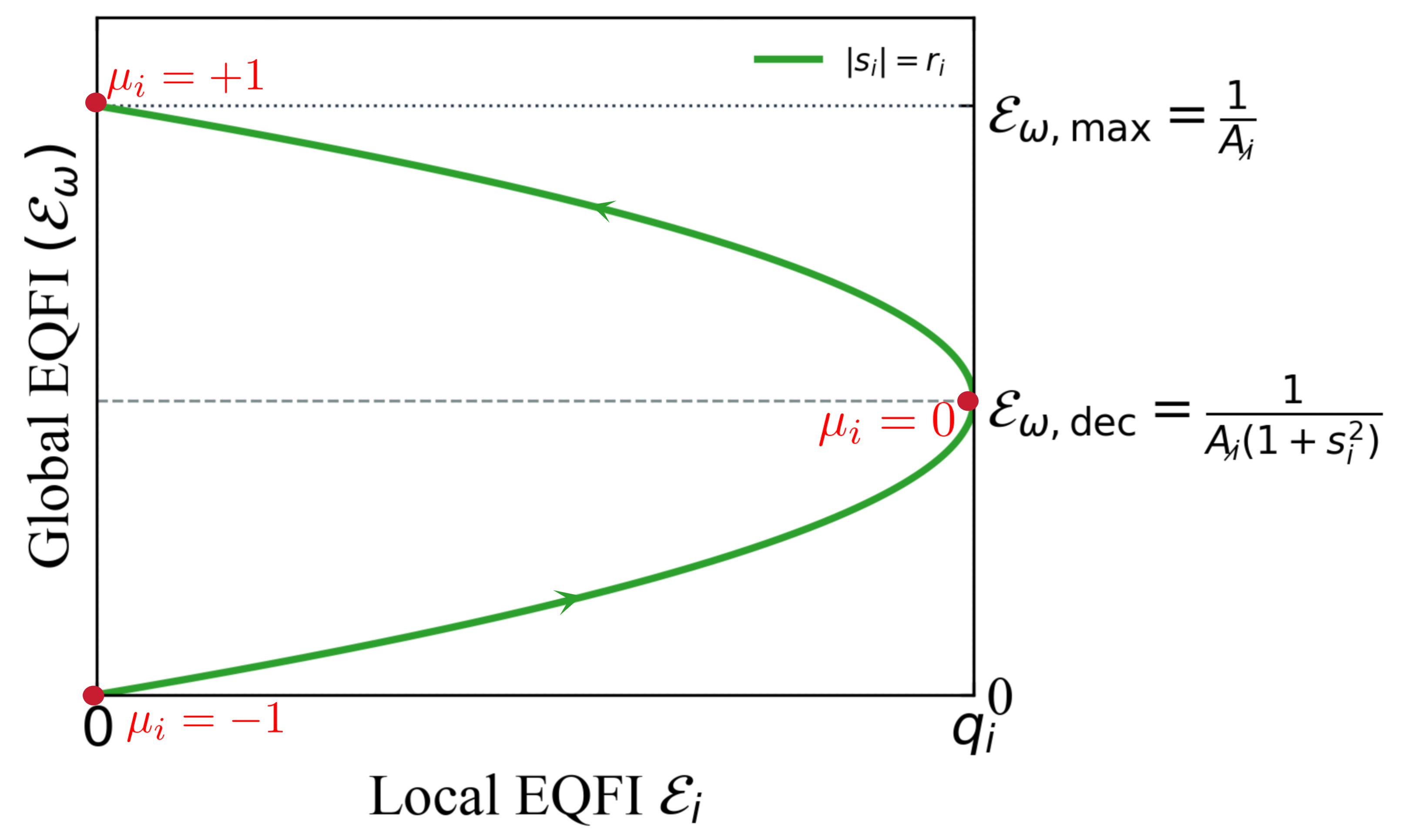}
\caption{\textit{Matched singular boundary.} For $|s_i|=r_i$, the aligned branch terminates at $\mathcal{E}_i=0$ and $\mathcal{E}_w=1/A_{\not{i}}$. The overcorrelated regime is absent, and the global target remains estimable at the singular endpoint.}
\label{fig: parameter-function-2}
\end{figure}

\subsection{Limiting Cases}

We finally consider the limits $s_i=0$ and $r_i=0$, for which the generic stationary-point formula is not applicable.

\textit{Zero node weight, $s_i=0$.---} In this case, $w_i=0$ and the global target is entirely supported on the complementary subnetwork:
\begin{equation}
\theta_w=\mathbf{w}_{\not{i}}^T\bm{\theta}_{\not{i}}.
\end{equation}
\cref{eq:global_map} reduces to
\begin{equation}\label{eq:zero_weight_sm}
\mathcal{E}_w(\mu_i)=\frac{1}{A_{\not{i}}}\frac{1-\mu_i^2}{1-(1-r_i^2)\mu_i^2}.
\end{equation}
For $r_i>0$, $\mathcal{E}_w$ decreases monotonically with $|\mu_i|$ and is maximized at the decoupled point:
\begin{equation}
\mathcal{E}_{w,\max}=\mathcal{E}_w(0)=\frac{1}{A_{\not{i}}}.
\end{equation}
Although node $i$ carries no target weight, correlations with that node introduce an additional nuisance coupling and reduce the global EQFI. When $r_i=1$,
\begin{equation}
\mathcal{E}_w(\mu_i)=\frac{1-\mu_i^2}{A_{\not{i}}}=\frac{\mathcal{E}_i(\mu_i)}{A_{\not{i}}q_i},
\end{equation}
so the local and global EQFIs are proportional along this correlation family.

\textit{Orthogonal overlap, $r_i=0$.---} Here,
\begin{equation}
D_{\not{i}}=\mathbf{w}_{\not{i}}^T\mathbf{Q}_{\not{i}}^{-1}\mathbf{v}_i=0,
\end{equation}
meaning that the complementary target direction is orthogonal to the correlation axis in the metric induced by $\mathbf{Q}_{\not{i}}^{-1}$. For $s_i\neq0$,
\begin{equation}\label{eq:zero_overlap_sm}
\mathcal{E}_w(\mu_i)=\frac{1}{A_{\not{i}}}\left(1+\frac{s_i^2}{1-\mu_i^2}\right)^{-1}.
\end{equation}
This quantity decreases monotonically with $|\mu_i|$, and its maximum is the decoupled value
\begin{equation}
\mathcal{E}_{w,\max}=\frac{1}{A_{\not{i}}(1+s_i^2)}.
\end{equation}
Correlations orthogonal to the target direction cannot enhance the global EQFI.

\textit{Fully decoupled target geometry, $s_i=r_i=0$.---} Combining the two limits gives
\begin{equation}
\mathcal{E}_w(\mu_i)=\frac{1}{A_{\not{i}}},
\end{equation}
independent of $\mu_i$. In this case, node $i$ has zero target weight and its correlation axis is orthogonal to the complementary target direction. Changing the node--subnetwork correlation therefore affects the local EQFI but leaves the global task unchanged.

\begin{figure}[t]
\centering
\includegraphics[width=\linewidth]{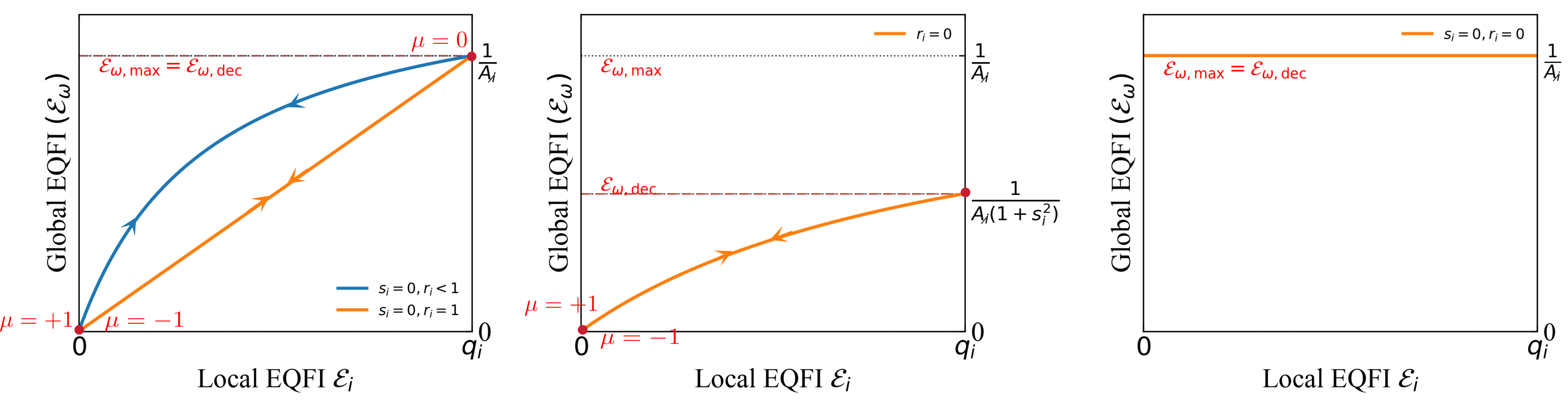}
\caption{\textit{Limiting phase-map geometries.} (a) Zero node weight, $s_i=0$ and $r_i>0$. (b) Orthogonal overlap, $r_i=0$ and $s_i\neq0$. (c) Fully decoupled target geometry, $s_i=r_i=0$.}
\label{fig: parameter-function-3}
\end{figure}

\begin{table}[t]
\centering
\caption{Behavior of the global EQFI for $s_i\ge0$. For $s_i<0$, the aligned branch is obtained by reversing the sign of $\mu_i$.}
\begin{tabular}{cccc}
\toprule
Condition & Behavior of $\mathcal{E}_w$ & Maximizer & Maximum \\
\midrule
$s_i,r_i>0$, $s_i\neq r_i$ & Increases, then decreases & $\mu_i^*$ & \cref{eq:subnetwork_ceiling_sm} or \cref{eq:bottleneck_branch_sm} \\
$s_i=r_i>0$ & Increases on aligned branch & $\mu_i=1$ & $1/A_{\not{i}}$ \\
$s_i=0$, $r_i>0$ & Decreases with $|\mu_i|$ & $\mu_i=0$ & $1/A_{\not{i}}$ \\
$r_i=0$, $s_i>0$ & Decreases with $|\mu_i|$ & $\mu_i=0$ & $1/[A_{\not{i}}(1+s_i^2)]$ \\
$s_i=r_i=0$ & Constant & Any $\mu_i$ & $1/A_{\not{i}}$ \\
\bottomrule
\end{tabular}
\label{tab:special_cases}
\end{table}

The phase map therefore provides a complete classification of the local--global precision relationship for an arbitrary node--subnetwork decomposition. Correlations enhance the global target only when they are aligned with its weight geometry and remain below the optimal strength. Beyond that optimum, additional QFIM correlation suppresses both local and global information, except at the matched singular boundary where local information vanishes while the global target remains estimable.

\section{Recursive Network Design}
\label{sec: recursive}

\subsection{Motivation and Recursive Principle}

We now address the inverse design problem: for fixed local QFI capacities ${q_i}$ and a prescribed target vector $\mathbf{w}$, how should the inter-node correlations be constructed to saturate the bottleneck bound? The node--subnetwork phase map shows that the optimal QFIM is generally nonunique. At the level of the recursive construction, the optimal family is described by one scalar overlap parameter at each of the $N-1$ node-addition steps,
\begin{equation}\label{eq: recursive_design_parameters}
\mathbf{r}=\left(r^{(1)},r^{(2)},\ldots,r^{(N-1)}\right).
\end{equation}
These overlap coordinates characterize the alignment between the correlation introduced at each step and the target direction of the existing subnetwork. Their realizable values depend on the dimension and support of the intermediate QFIMs, and different correlation directions may also produce the same overlap.

The construction follows directly from the phase map. According to the Corollary I in the main text, adding a node to an existing subnetwork preserves the subnetwork ceiling whenever the node-dependent term vanishes, namely when $|s|\le r$ and $\mu=s/r$. We therefore begin with a one-node network that already attains the bottleneck precision and add the remaining nodes one by one, choosing the correlation at each step so that the global EQFI remains unchanged.

We first sort and relabel the nodes according to their intrinsic weighted capacities,
\begin{equation}\label{eq: recursive_capacity_order}
\kappa_1\ge\kappa_2\ge\cdots\ge\kappa_N,\qquad \kappa_i\equiv\frac{q_i}{w_i^2},
\end{equation}
where $\kappa_i=+\infty$ when $w_i=0$. Node $N$ is then a bottleneck node and fixes the optimal precision,
\begin{equation}\label{eq: recursive_bottleneck}
\E_{w,\max}=\kappa_N=\frac{q_N}{w_N^2}=\min_i\kappa_i.
\end{equation}

Let $\mathbf{Q}^{(k)}$ denote the QFIM of the network containing nodes $k,k+1,\ldots,N$, and let $\mathbf{w}^{(k)}=(w_k,w_{k+1},\ldots,w_N)^T$ denote its target-weight vector. At the $k$th recursive step, node $k$ is added to the already constructed network $\mathbf{Q}^{(k+1)}$:
\begin{equation}\label{eq: recursive_block}
\mathbf{Q}^{(k)}=\begin{pmatrix}q_k&(\mathbf{c}^{(k)})^T\\ \mathbf{c}^{(k)}&\mathbf{Q}^{(k+1)}\end{pmatrix},\qquad \mathbf{w}^{(k)}=\begin{pmatrix}w_k\\ \mathbf{w}^{(k+1)}\end{pmatrix}.
\end{equation}
Here, $\mathbf{c}^{(k)}$ is the correlation vector introduced when constructing the $k$-to-$N$ network. For any real vector $\mathbf{x}$, $||\mathbf{x}||\equiv\sqrt{\mathbf{x}^T\mathbf{x}}$ denotes the Euclidean norm.

To distinguish the recursive construction from the arbitrary node--subnetwork decomposition considered previously, we use $s^{(k)}$, $r^{(k)}$, and $\mu^{(k)}$ for the phase-map parameters associated with the formation of $\mathbf{Q}^{(k)}$. By contrast, $A^{(k+1)}$, $B^{(k+1)}$, and $D^{(k+1)}$ characterize the existing subnetwork $\mathbf{Q}^{(k+1)}$.

The recursion preserves two properties:
\begin{equation}\label{eq: recursive_invariant}
\mathbf{w}^{(k)}\in\operatorname{range}(\mathbf{Q}^{(k)}),\qquad A^{(k)}\equiv(\mathbf{w}^{(k)})^T(\mathbf{Q}^{(k)})^+\mathbf{w}^{(k)}=\frac{1}{\kappa_N}.
\end{equation}
The first condition ensures that the target remains estimable even when the QFIM is singular, while the second gives $\E_w^{(k)}=1/A^{(k)}=\kappa_N$.

The construction is initialized with the bottleneck node alone. Since $\mathbf{Q}^{(N)}=q_N$ and $\mathbf{w}^{(N)}=w_N$, one immediately has $A^{(N)}=w_N^2/q_N=1/\kappa_N$. The remaining nodes are then incorporated successively according to
\begin{equation}
\mathbf{Q}^{(N)}\xrightarrow{\mathbf{c}^{(N-1)}}\mathbf{Q}^{(N-1)}\xrightarrow{\mathbf{c}^{(N-2)}}\cdots\xrightarrow{\mathbf{c}^{(1)}}\mathbf{Q}^{(1)}.
\end{equation}

\subsection{General Optimal Family with Arbitrary $r^{(k)}$}

Assume that the subnetwork $\mathbf{Q}^{(k+1)}$ has already been constructed and satisfies the recursive invariant in \cref{eq: recursive_invariant}. We incorporate node $k$ through
\begin{equation}\label{eq: recursive_block_general}
\mathbf{Q}^{(k)}=\begin{pmatrix}q_k&(\mathbf{c}^{(k)})^T\\ \mathbf{c}^{(k)}&\mathbf{Q}^{(k+1)}\end{pmatrix},\qquad \mathbf{c}^{(k)}=c^{(k)}\mathbf{v}^{(k+1)},
\end{equation}
where $c^{(k)}\in\mathbb{R}$ is a signed amplitude and $\mathbf{v}^{(k+1)}\in\operatorname{range}(\mathbf{Q}^{(k+1)})$ is a Euclidean unit vector. We orient $\mathbf{v}^{(k+1)}$ such that $D^{(k+1)}\ge0$, so the sign of the correlation is carried entirely by $c^{(k)}$.

The geometry of the existing subnetwork is characterized by
\begin{equation}\label{eq: recursive_ABD}
A^{(k+1)}\equiv(\mathbf{w}^{(k+1)})^T(\mathbf{Q}^{(k+1)})^+\mathbf{w}^{(k+1)},\qquad B^{(k+1)}\equiv(\mathbf{v}^{(k+1)})^T(\mathbf{Q}^{(k+1)})^+\mathbf{v}^{(k+1)},\qquad D^{(k+1)}\equiv(\mathbf{w}^{(k+1)})^T(\mathbf{Q}^{(k+1)})^+\mathbf{v}^{(k+1)}.
\end{equation}
Correspondingly, the phase-map parameters associated with the enlarged $k$-to-$N$ network are
\begin{equation}\label{eq: recursive_phase_parameters}
s^{(k)}\equiv\frac{w_k/\sqrt{q_k}}{\sqrt{A^{(k+1)}}},\qquad r^{(k)}\equiv\frac{D^{(k+1)}}{\sqrt{A^{(k+1)}B^{(k+1)}}},\qquad \mu^{(k)}\equiv c^{(k)}\sqrt{\frac{B^{(k+1)}}{q_k}}.
\end{equation}

We first show how to construct a correlation direction with any prescribed overlap $r^{(k)}$. Let $\mathbf{M}^{(k+1)}\equiv(\mathbf{Q}^{(k+1)})^+$. When $\operatorname{rank}(\mathbf{Q}^{(k+1)})\ge2$, choose a vector $\mathbf{u}^{(k+1)}\in\operatorname{range}(\mathbf{Q}^{(k+1)})$ that is orthogonal to $\mathbf{w}^{(k+1)}$ in the metric induced by $\mathbf{M}^{(k+1)}$ and normalized in the same metric:
\begin{equation}
(\mathbf{w}^{(k+1)})^T\mathbf{M}^{(k+1)}\mathbf{u}^{(k+1)}=0,\qquad (\mathbf{u}^{(k+1)})^T\mathbf{M}^{(k+1)}\mathbf{u}^{(k+1)}=1.
\end{equation}

Then define
\begin{equation}\label{eq: recursive_direction_general}
\widetilde{\mathbf{v}}^{(k+1)}=r^{(k)}\frac{\mathbf{w}^{(k+1)}}{\sqrt{A^{(k+1)}}}+\sqrt{1-[r^{(k)}]^2}\mathbf{u}^{(k+1)},\qquad \mathbf{v}^{(k+1)}=\frac{\widetilde{\mathbf{v}}^{(k+1)}}{|\widetilde{\mathbf{v}}^{(k+1)}|}.
\end{equation}
The Euclidean normalization does not change the metric angle. Indeed,
\begin{equation}\label{eq: recursive_direction_verification}
B^{(k+1)}=\frac{1}{|\widetilde{\mathbf{v}}^{(k+1)}|^2},\qquad D^{(k+1)}=\frac{r^{(k)}\sqrt{A^{(k+1)}}}{|\widetilde{\mathbf{v}}^{(k+1)}|},
\end{equation}
and hence $D^{(k+1)}/\sqrt{A^{(k+1)}B^{(k+1)}}=r^{(k)}$, as required. Thus, every overlap in $[0,1]$ can be realized whenever the support of $\mathbf{Q}^{(k+1)}$ has dimension at least two. 

If $\operatorname{rank}(\mathbf{Q}^{(k+1)})=1$, every nonzero direction in its range is parallel to $\mathbf{w}^{(k+1)}$, and only $r^{(k)}=1$ is available.

The recursive invariant gives $A^{(k+1)}=1/\kappa_N$, so
\begin{equation}
s^{(k)}=w_k\sqrt{\frac{\kappa_N}{q_k}},\qquad |s^{(k)}|=\sqrt{\frac{\kappa_N}{\kappa_k}}.
\end{equation}
The phase-map ceiling is therefore attainable whenever the prescribed overlap satisfies
\begin{equation}\label{eq: recursive_r_admissible}
r^{(k)}\in\left[\sqrt{\frac{\kappa_N}{\kappa_k}},1\right].
\end{equation}
For $w_k\neq0$, this interval has a strictly positive lower endpoint, and hence $D^{(k+1)}>0$. We then choose the signed amplitude $c^{(k)}=w_k/D^{(k+1)}$, or equivalently
\begin{equation}\label{eq: recursive_correlation_general}
\mathbf{c}^{(k)}=\frac{w_k}{D^{(k+1)}}\mathbf{v}^{(k+1)}=\frac{w_k}{r^{(k)}}\sqrt{\frac{\kappa_N}{B^{(k+1)}}}\mathbf{v}^{(k+1)}.
\end{equation}
Because $D^{(k+1)}\ge0$, the sign of $c^{(k)}$ is the sign of $w_k$. This construction gives $\mu^{(k)}=s^{(k)}/r^{(k)}$ and enforces the cancellation condition
\begin{equation}\label{eq: recursive_cancellation}
(\mathbf{c}^{(k)})^T(\mathbf{Q}^{(k+1)})^+\mathbf{w}^{(k+1)}=w_k.
\end{equation}
Equation~\eqref{eq: recursive_cancellation} is the matrix form of the phase-map ceiling condition: the correlation-induced contribution from the existing subnetwork exactly matches the target weight of the newly added node. If $w_k=0$, the ceiling is attained simply by choosing $\mathbf{c}^{(k)}=\mathbf{0}$.

The same admissibility condition also guarantees that the enlarged QFIM remains positive semidefinite. Since the constructed correlation vector lies in $\operatorname{range}(\mathbf{Q}^{(k+1)})$, the generalized Schur complement is
\begin{equation}\label{eq: recursive_schur_general}
\E_k^{(k)}=q_k-(\mathbf{c}^{(k)})^T(\mathbf{Q}^{(k+1)})^+\mathbf{c}^{(k)}=q_k\left(1-\frac{\kappa_N}{\kappa_k[r^{(k)}]^2}\right)\ge0.
\end{equation}
Thus, \cref{eq: recursive_r_admissible} is precisely the positivity condition for the newly added node.

Finally, we verify that the global precision is preserved. Let $\mathbf{z}^{(k+1)}=(\mathbf{Q}^{(k+1)})^+\mathbf{w}^{(k+1)}$. The induction assumption implies $\mathbf{Q}^{(k+1)}\mathbf{z}^{(k+1)}=\mathbf{w}^{(k+1)}$, while \cref{eq: recursive_cancellation} gives $(\mathbf{c}^{(k)})^T\mathbf{z}^{(k+1)}=w_k$. Consequently,
\begin{equation}\label{eq: recursive_preservation}
\mathbf{Q}^{(k)}\begin{pmatrix}0\\ \mathbf{z}^{(k+1)}\end{pmatrix}=\mathbf{w}^{(k)},\qquad A^{(k)}=A^{(k+1)}=\frac{1}{\kappa_N}.
\end{equation}
The target therefore remains estimable and its global EQFI is preserved at $\E_w^{(k)}=\kappa_N$. This completes the inductive step.

The sequence $r^{(1)},r^{(2)},\ldots,r^{(N-1)}$ consequently labels a family of bottleneck-saturating recursive designs. When the intermediate subnetwork has support dimension at least two, every value in the admissible interval of \cref{eq: recursive_r_admissible} can be realized through \cref{eq: recursive_direction_general}. Different choices of the auxiliary vector $\mathbf{u}^{(k+1)}$ can nevertheless generate distinct correlation directions with the same overlap, so the scalar sequence $r^{(k)}$ does not uniquely specify the complete QFIM.

At the first recursive step, the existing subnetwork contains only the bottleneck node and is one-dimensional; therefore, $r^{(N-1)}=1$ necessarily. At the lower endpoint $r^{(k)}=\sqrt{\kappa_N/\kappa_k}$, the generalized Schur complement vanishes and node $k$ becomes locally inestimable, while the global EQFI remains equal to $\kappa_N$. For larger overlaps, the newly added node retains finite local estimability.

\subsection{Canonical Maximally Aligned Construction: $r^{(k)}=1$}

We now select a canonical representative of the optimal family by setting
\begin{equation}\label{eq: recursive_all_aligned}
r^{(k)}=1,\qquad k=N-1,N-2,\ldots,1.
\end{equation}
Thus, at every recursive step, the new correlation direction is chosen to be maximally aligned with the target direction of the existing subnetwork:
\begin{equation}\label{eq: recursive_direction_aligned}
    \mathbf{v}^{(k+1)}=\frac{\mathbf{w}^{(k+1)}}{\|\mathbf{w}^{(k+1)}\|}.
\end{equation}
For this direction,
\begin{equation}
    B^{(k+1)}=\frac{A^{(k+1)}}{\|\mathbf{w}^{(k+1)}\|^2},\qquad
    D^{(k+1)}=\frac{A^{(k+1)}}{\|\mathbf{w}^{(k+1)}\|},
\end{equation}
which verifies $r^{(k)}=1$. These relations remain valid for a singular subnetwork upon replacing the inverse of its QFIM by the Moore--Penrose pseudoinverse.

Using $A^{(k+1)}=1/\kappa_N$, the general prescription in \cref{eq: recursive_correlation_general} simplifies to
\begin{equation}\label{eq: recursive_correlation_aligned}
    \mathbf{c}^{(k)}=\kappa_N w_k\mathbf{w}^{(k+1)}.
\end{equation}
The corresponding phase-map strength and local EQFI are
\begin{equation}
    \mu^{(k)}=s^{(k)}=w_k\sqrt{\frac{\kappa_N}{q_k}}, \qquad
    \E_k^{(k)}=q_k-\kappa_Nw_k^2.
\end{equation}
Since the general correlation penalty scales as
$\kappa_Nw_k^2/[r^{(k)}]^2$, the choice $r^{(k)}=1$ minimizes the correlation-induced reduction of the local EQFI required to reach the subnetwork ceiling.

Iterating \cref{eq: recursive_correlation_aligned} fixes every off-diagonal element of the final QFIM to
$Q_{ij}=\kappa_Nw_iw_j$. The complete network therefore takes the closed form
\begin{equation}\label{eq: recursive_closed_form}
    \mathbf{Q}=\kappa_N\mathbf{w}\mathbf{w}^T+\operatorname{diag}\left(q_1-\kappa_Nw_1^2,\ldots, q_N-\kappa_Nw_N^2\right).
\end{equation}
Every diagonal correction is nonnegative because
$\kappa_i\geq\kappa_N$, so the resulting QFIM is manifestly positive semidefinite.

The closed form in \cref{eq: recursive_closed_form} also includes networks with a degenerate bottleneck. Let $\mathcal B=\left\{i:\kappa_i=\kappa_N\right\}$ denote the set of bottleneck nodes. If $|\mathcal B|=1$, the rank-one term restores the direction missing from the diagonal part, and $\mathbf Q$ is positive definite. If $|\mathcal B|\geq2$, however, the QFIM is singular, with
\begin{equation}\label{eq: canonical_kernel}
    \ker(\mathbf Q)=\left\{\mathbf{x}:x_i=0\ \text{for}\ i\notin\mathcal B,\quad\sum_{i\in\mathcal B}w_ix_i=0\right\},
\end{equation}
and hence
\begin{equation}
    \operatorname{rank}(\mathbf Q)=N-|\mathcal B|+1.
\end{equation}
Nevertheless, the target direction remains estimable. Indeed, for any $j\in\mathcal B$, $\mathbf Q\frac{\mathbf e_j}{\kappa_Nw_j}= \mathbf w$, which proves that $\mathbf w\in\operatorname{range}(\mathbf Q)$. Because any two solutions of $\mathbf Q\mathbf z=\mathbf w$ differ by a vector in $\ker(\mathbf Q)$, which is orthogonal to $\mathbf w$, it gives
\begin{equation}
    \mathbf w^T\mathbf Q^+\mathbf w=\frac{1}{\kappa_N},\qquad
    \E_w=\left(\mathbf w^T\mathbf Q^+\mathbf w\right)^{-1}=\kappa_N=\min_i\kappa_i.
\end{equation}
Thus, bottleneck saturation persists even when the optimal QFIM is singular.

When an intermediate subnetwork contains more than one bottleneck node, its QFIM is likewise singular. The recursive construction then proceeds by interpreting all inverse-based quantities $A^{(k)}$, $B^{(k)}$, and $D^{(k)}$ using the corresponding Moore--Penrose pseudoinverse. This extension is well defined because
\begin{equation}
    \mathbf c^{(k)}=\kappa_Nw_k\mathbf w^{(k+1)}\in\operatorname{range}\!\left(\mathbf Q^{(k+1)}\right),
\end{equation}
so the range condition required by the generalized Schur complement is automatically satisfied. The aligned relations above and the invariant
$A^{(k)}=1/\kappa_N$ therefore remain valid at every recursive step.

The recursive phase-map construction consequently identifies a family of optimal QFIMs under fixed local capacities, including both regular and singular optima. The overlaps $r^{(k)}$ control the distribution and strength of the required correlations, whereas the choice $r^{(k)}=1$ at every level selects the maximally aligned representative and yields the compact closed-form QFIM in \cref{eq: recursive_closed_form}.

\section{Permutation-Invariant Networks}
\label{sec:permutation_invariant}

\subsection{Exact Global EQFI and Sensitivity to Weight Asymmetry}

We consider an analytically solvable class of quantum sensor networks whose QFIM is invariant under permutations of the nodes. Such a structure arises naturally for permutation-symmetric multipartite probes, including GHZ, W, and collectively spin-squeezed states with identical local generators. The QFIM takes the form
\begin{equation}\label{eq:permutation_qfim}
\mathbf Q = \begin{pmatrix} q & c & c & \cdots & c \\ c & q & c & \cdots & c \\ c & c & q & \cdots & c \\ \vdots & \vdots & \vdots & \ddots & \vdots \\ c & c & c & \cdots & q \end{pmatrix} = (q-c)\mathbf I_N+c\mathbf J_N,
\end{equation}
where $\mathbf I_N$ is the identity matrix and $\mathbf J_N$ is the all-one matrix. The diagonal element $q$ is the single-parameter local QFI of each node, whereas $c$ quantifies the uniform pairwise QFIM correlation. Introducing $\gamma=c/q$, positivity of $\mathbf Q$ requires $-\frac{1}{N-1}\le\gamma\le1$.

To connect this model with the general node--subnetwork phase map, we isolate node $1$ and write $m=N-1$. The complementary subnetwork and the node--subnetwork correlation vector are
\begin{equation}
\mathbf Q_{\not 1}=(q-c)\mathbf I_m+c\mathbf J_m,\qquad \mathbf c_1=c\mathbf 1_m,
\end{equation}
where $\mathbf 1_m$ is the $m$-dimensional all-one vector. For $-1/(N-1)<\gamma<1$, the inverse of the complementary QFIM is
\begin{equation}\label{eq:permutation_subinverse}
\mathbf Q_{\not 1}^{-1}=\frac{1}{q-c}\mathbf I_m-\frac{c}{(q-c)[q+(N-2)c]}\mathbf J_m.
\end{equation}

Because $\mathbf c_1$ points along the uniform direction, we choose $\mathbf v=\frac{\mathbf 1_m}{\sqrt m}$ and $\mathbf c_1=\sqrt m\,c\,\mathbf v$. 
Defining 
\begin{equation}
W_{\not 1}\equiv\sum_{i=2}^Nw_i,\qquad S_{\not 1}\equiv\sum_{i=2}^Nw_i^2,
\end{equation}
the geometric quantities entering the phase map become
\begin{equation}\label{eq:permutation_ABD}
\begin{aligned}
B_{\not 1}&=\mathbf v^T\mathbf Q_{\not 1}^{-1}\mathbf v=\frac{1}{q+(N-2)c},\\
A_{\not 1}&=\mathbf w_{\not 1}^T\mathbf Q_{\not 1}^{-1}\mathbf w_{\not 1}=\frac{S_{\not 1}}{q-c}-\frac{cW_{\not 1}^2}{(q-c)[q+(N-2)c]},\\
D_{\not 1}&=\mathbf w_{\not 1}^T\mathbf Q_{\not 1}^{-1}\mathbf v=\frac{W_{\not 1}}{\sqrt{N-1}[q+(N-2)c]}.
\end{aligned}
\end{equation}
Accordingly,
\begin{equation}\label{eq:permutation_phase_parameters}
r=\frac{D_{\not 1}}{\sqrt{A_{\not 1}B_{\not 1}}},\qquad
s=\frac{w_1}{\sqrt{qA_{\not 1}}},\qquad
\mu=\frac{\sqrt{N-1}\,c}{\sqrt{q[q+(N-2)c]}}.
\end{equation}
In terms of the normalized correlation $\gamma$, the local information-loss parameter is
\begin{equation}\label{eq:permutation_mu}
\mu^2=\frac{(N-1)\gamma^2}{1+(N-2)\gamma}.
\end{equation}
The general phase map in \cref{eq:local_map,eq:global_map} therefore applies directly:
\begin{equation}
\mathcal E_1=q(1-\mu^2),\qquad
\mathcal E_w=\frac{1}{A_{\not 1}}\left[1+\frac{(s-r\mu)^2}{1-\mu^2}\right]^{-1}.
\end{equation}

We next derive the global EQFI in closed form. The inverse of the full QFIM is
\begin{equation}\label{eq:permutation_full_inverse}
\mathbf Q^{-1}=\frac{1}{q-c}\mathbf I_N-\frac{c}{(q-c)[q+(N-1)c]}\mathbf J_N.
\end{equation}
We focus on normalized nonnegative weights satisfying $\sum_iw_i=1$ and define their inverse participation ratio
\begin{equation}
S\equiv\sum_{i=1}^Nw_i^2,\qquad \frac{1}{N}\le S\le1.
\end{equation}
Using $\mathbf w^T\mathbf J_N\mathbf w=(\sum_iw_i)^2=1$, we obtain
\begin{equation}
\mathbf w^T\mathbf Q^{-1}\mathbf w=\frac{S}{q-c}-\frac{c}{(q-c)[q+(N-1)c]}.
\end{equation}
The exact global EQFI is consequently
\begin{equation}\label{eq:permutation_global_eqfi}
\mathcal E_w=q\,\frac{(1-\gamma)[1+(N-1)\gamma]}{S[1+(N-1)\gamma]-\gamma}.
\end{equation}

The correlation-free case $\gamma=0$ provides the SQL reference
\begin{equation}\label{eq:permutation_sql}
\mathcal E_{w,\mathrm{SQL}}=\frac{q}{S}.
\end{equation}
Because the QFIM itself is permutation invariant, any mismatch between its symmetry and the sensing task is encoded entirely in the weight profile. We quantify the uniformity of this profile by
\begin{equation}\label{eq:weight_uniformity}
\alpha_S\equiv\frac{1-S}{S}=\frac{1}{S}-1,\qquad
\widetilde{\alpha}_S\equiv\frac{\alpha_S}{N-1}=\frac{1-S}{S(N-1)}.
\end{equation}
The normalized factor satisfies $0\le\widetilde{\alpha}_S\le1$. The limit $\widetilde{\alpha}_S=0$ corresponds to a target concentrated on a single node, whereas $\widetilde{\alpha}_S=1$ corresponds to the uniform target $w_i=1/N$.

Dividing \cref{eq:permutation_global_eqfi} by \cref{eq:permutation_sql} gives
\begin{equation}\label{eq:permutation_eqfi_ratio}
\widetilde{\E}_w \equiv \frac{\E_w}{\E_{w,\mathrm{SQL}}}= 1+(N-1)\gamma \frac{\widetilde{\alpha}_S-\gamma}{1+\gamma R},
\end{equation}
with $R=(N-1)(1-\widetilde{\alpha}_S)-1$. 

Optimizing this expression over the physically allowed correlation strength yields
\begin{equation}\label{eq:permutation_optimal_gamma}
\gamma_\star=\frac{\widetilde{\alpha}_S}{1+T_S},\qquad
T_S\equiv\sqrt{(1-\widetilde{\alpha}_S)[1+(N-1)\widetilde{\alpha}_S]}.
\end{equation}
At this optimum,
\begin{equation}\label{eq:permutation_max_ratio}
\frac{\mathcal E_{w,\max}}{\mathcal E_{w,\mathrm{SQL}}}
=1+(N-1)\gamma_\star^2
=1+(N-1)\frac{\widetilde{\alpha}_S^2}{(1+T_S)^2}.
\end{equation}

For the uniform target, $\widetilde{\alpha}_S=1$, so $T_S=0$ and $\gamma_\star=1$. Therefore,
\begin{equation}
\frac{\mathcal E_{w,\max}}{\mathcal E_{w,\mathrm{SQL}}}=N,\qquad
\mathcal E_{w,\mathrm{SQL}}=qN,\qquad
\mathcal E_{w,\max}=qN^2.
\end{equation}
The rank-one limit $\gamma=1$ thus produces Heisenberg scaling when its unique estimable direction is exactly matched to the uniform target.

The fragility of this enhancement can be quantified by writing
\begin{equation}
\widetilde{\alpha}_S=1-\delta.
\end{equation}
Then
\begin{equation}
T_S=\sqrt{\delta[N-(N-1)\delta]}.
\end{equation}
If $\delta=\Theta(1)$, \cref{eq:permutation_max_ratio} gives
\begin{equation}
\frac{\mathcal E_{w,\max}}{\mathcal E_{w,\mathrm{SQL}}}\sim\frac{1}{\delta}=O(1).
\end{equation}
Since $\mathcal E_{w,\mathrm{SQL}}=O(N)$ in this regime, the optimized EQFI retains only SQL scaling, $\mathcal E_{w,\max}=O(N)$.

By contrast, if $\delta=O(1/N)$, then
\begin{equation}
\frac{\mathcal E_{w,\max}}{\mathcal E_{w,\mathrm{SQL}}}
\sim\frac{N}{(1+\sqrt{N\delta})^2},
\end{equation}
and hence
\begin{equation}
\mathcal E_{w,\max}\sim\frac{qN^2}{(1+\sqrt{N\delta})^2}.
\end{equation}
An $O(1/N)$ symmetry mismatch therefore changes only the asymptotic prefactor, whereas an $O(1)$ mismatch destroys the quadratic scaling. Preserving Heisenberg scaling in a permutation-invariant network consequently requires the target weights to approach uniformity at least as rapidly as $1-\widetilde{\alpha}_S=O(1/N)$.

\subsection{Comparison with the Projected QFI}

Several studies characterize distributed sensing through the QFI associated with a prescribed one-dimensional parameter direction~\cite{hassani2025privacy,farokhi2026precision}. For the target $\theta_w=\mathbf w^T\bm{\theta}$, this projected QFI is
\begin{equation}\label{eq:projected_qfi_permutation}
F_w\equiv\frac{\mathbf w^T\mathbf Q\mathbf w}{(\mathbf w^T\mathbf w)^2}.
\end{equation}
It corresponds to the one-parameter embedding $\bm{\theta}=\theta_w\mathbf w/(\mathbf w^T\mathbf w)$, for which all parameter directions orthogonal to $\mathbf w$ are assumed known or absent. The associated generator is $(\mathbf w^T\mathbf w)^{-1}\sum_iw_iH_i$.

The EQFI instead describes the genuinely multiparameter setting in which the orthogonal parameter combinations remain unknown and act as nuisance parameters:
\begin{equation}
\mathcal E_w=\left(\mathbf w^T\mathbf Q^{-1}\mathbf w\right)^{-1}.
\end{equation}
These two quantities satisfy $F_w\ge\mathcal E_w$, with equality precisely when $\mathbf w$ is an eigenvector of $\mathbf Q$.

For the permutation-invariant QFIM in \cref{eq:permutation_qfim},
\begin{equation}
\mathbf w^T\mathbf Q\mathbf w=q[(1-\gamma)S+\gamma],
\end{equation}
and therefore
\begin{equation}\label{eq:permutation_projected_qfi}
F_w=q\,\frac{(1-\gamma)S+\gamma}{S^2}.
\end{equation}
Relative to the SQL reference,
\begin{equation}\label{eq:permutation_projected_ratio}
\frac{F_w}{\mathcal E_{w,\mathrm{SQL}}}
=1+\gamma\left(\frac{1}{S}-1\right)
=1+\gamma\alpha_S.
\end{equation}
Thus, for nonnegative correlations, the projected QFI increases linearly with $\gamma$.

By contrast, the exact EQFI obeys the nonlinear relation
\begin{equation}
\frac{\mathcal E_w}{\mathcal E_{w,\mathrm{SQL}}}
=\frac{(1-\gamma)[1+(N-1)\gamma]}{1+\gamma(N-2-\alpha_S)},
\end{equation}
which is generally nonmonotonic. Its maximum occurs at the finite correlation strength $\gamma_\star$ in \cref{eq:permutation_optimal_gamma}, except for the perfectly uniform target, for which $\gamma_\star=1$.

For uniform weights, $S=1/N$ and $\alpha_S=N-1$. The vector $\mathbf w$ is then the symmetric eigenvector of $\mathbf Q$, and
\begin{equation}
F_w=\mathcal E_w=qN[1+(N-1)\gamma].
\end{equation}
The projected single-parameter model and the nuisance-aware multiparameter model therefore coincide in this perfectly aligned case.

For every nonuniform normalized weight profile, $S>1/N$, the distinction becomes sharp in the rank-one limit $\gamma\to1$. The projected QFI remains finite and reaches
\begin{equation}
\lim_{\gamma\to1}\frac{F_w}{\mathcal E_{w,\mathrm{SQL}}}=\frac{1}{S},
\end{equation}
whereas the exact EQFI vanishes:
\begin{equation}
\lim_{\gamma\to1}\frac{\mathcal E_w}{\mathcal E_{w,\mathrm{SQL}}}=0.
\end{equation}
At $\gamma=1$, the QFIM is $\mathbf Q=q\mathbf J_N$ and has range $\operatorname{span}(\mathbf 1_N)$. A nonuniform target vector does not belong to this range and is therefore unestimable, despite its nonzero projected QFI.

The local--global phase map provides the same interpretation. From \cref{eq:permutation_mu}, $\gamma\to1$ implies $|\mu|\to1$ and hence $\mathcal E_1\to0$. For a nonuniform target, the singular correlation direction is not globally matched, so the global EQFI collapses as well. The projected QFI does not register this collapse because it removes the nuisance directions before evaluating the information.

Consequently, $F_w$ remains useful for a controlled one-parameter sensing model and for characterizing information along prescribed directions. It fails to, however, diagnose the nonmonotonic local--global trade-off, identify the onset of the overcorrelated regime, or determine the singular matching condition required to preserve global estimability. These features depend essentially on the inverse-QFIM geometry captured by the EQFI.

\section{Intrinsic Local and Functional Privacy in Quantum Sensor Networks}
\label{sec: privacy}

Distributed quantum sensing estimates a prescribed collective feature,
$\theta_w=\mathbf w^T\bm{\theta}$, without necessarily reconstructing the complete field profile $\bm{\theta}$. In many applications, however, the individual parameters may contain sensitive information. This motivates the question of whether the desired global quantity can remain estimable while every individual parameter is inaccessible. Throughout this section, privacy refers to QFIM-level non-identifiability and does not, by itself, constitute a cryptographic security guarantee.

\subsection{Intrinsic Local Privacy}

We define \textit{intrinsic local privacy} by
\begin{equation}\label{eq:local_privacy_definition}
    \E_i=0\quad\forall i,\qquad \E_w>0.
\end{equation}
The first condition means that no individual parameter $\theta_i$ admits a locally unbiased estimator with finite variance when all remaining parameters are unknown, whereas the second ensures that the prescribed sensing task remains operational.

For a singular QFIM, a scalar target $\theta_{\mathbf a}=\mathbf a^T\bm{\theta}$ is locally estimable if and only if $\mathbf a\in\operatorname{range}(\mathbf Q)$. Its EQFI is therefore
\begin{equation}\label{eq:singular_eqfi_privacy_sm}
\E_{\mathbf a}=
\begin{cases}
[\mathbf a^T\mathbf Q^+\mathbf a]^{-1},
& \mathbf a\in\operatorname{range}(\mathbf Q)\\
0, & \mathbf a\notin\operatorname{range}(\mathbf Q).
\end{cases}
\end{equation}
Applying this criterion to $\mathbf a=\mathbf e_i$ and $\mathbf a=\mathbf w$ gives
\begin{equation}\label{eq:privacy_range_sm}
\E_i=0\ \forall i,\quad \E_w>0
\quad\Longleftrightarrow\quad
\mathbf e_i\notin\operatorname{range}(\mathbf Q)\ \forall i,\quad
\mathbf w\in\operatorname{range}(\mathbf Q).
\end{equation}
Thus, every local coordinate direction lies outside the estimable subspace, while the global target direction remains inside it.

Equation~\eqref{eq:privacy_range_sm} also constrains the QFIM nullity
$r_{\mathrm K}\equiv\dim\ker(\mathbf Q)$. Intrinsic local privacy requires $\mathbf Q$ to be singular, because a full-rank QFIM has
$\operatorname{range}(\mathbf Q)=\mathbb R^N$ and therefore contains every $\mathbf e_i$. At the same time, $\E_w>0$ requires $\mathbf Q\neq0$. Consequently,
\begin{equation}\label{eq:privacy_nullity_range_sm}
1\le r_{\mathrm K}\le N-1,
\qquad
1\le\operatorname{rank}(\mathbf Q)=N-r_{\mathrm K}\le N-1.
\end{equation}
The nullity alone is not sufficient for privacy: the kernel must also be oriented such that every $\mathbf e_i$ has a nonzero component in $\ker(\mathbf Q)$, while $\mathbf w\perp\ker(\mathbf Q)$.

By the phase-map characterization, intrinsic local privacy corresponds to the simultaneous matched-singular conditions
\begin{equation}\label{eq:privacy_phase_sm}
    |\mu_i|=1,\qquad s_i=\mu_i r_i,\qquad\forall i.
\end{equation}

When $\mathbf Q_{\not i}$ is itself singular, as may occur for $r_{\mathrm K}>1$, the inverse-based phase-map variables are understood through the canonical regularization $\mathbf Q_\epsilon=\mathbf Q+\epsilon\mathbf I_N$ and let $\epsilon\to 0$. For nondegenerate participating nodes, the phase-map statement then means
\begin{equation}\label{eq:privacy_phase_limit_sm}
    |\mu_i(\epsilon)|\longrightarrow1,\qquad
    s_i(\epsilon)-\mu_i(\epsilon)r_i(\epsilon)\longrightarrow0.
\end{equation}
More precisely, we have:
\begin{equation}
    \frac{[s_i(\epsilon)-\mu_i(\epsilon)r_i(\epsilon)]^2} {1-\mu_i^2(\epsilon)} \longrightarrow0.
\end{equation}
The global EQFI gives: 
\begin{equation}
    \E_w(\epsilon)= \frac{1}{A_{\not i}(\epsilon)}\left[1+\frac{[s_i(\epsilon)-\mu_i(\epsilon)r_i(\epsilon)]^2}{1-\mu_i(\epsilon)^2}\right]^{-1}\to \frac{1}{A_{\not i}}.
\end{equation}
It shows explicitly how the local information vanishes without causing the global EQFI to collapse. The equalities in \cref{eq:privacy_phase_sm} are therefore the singular-boundary shorthand for these regularized limits.

\subsection{Functional Privacy}

A stronger notion, \textit{functional privacy}, requires that the sensing model reveal information only about the authorized function $\theta_w$~\cite{shettell2022private,hassani2025privacy,namkung2026universal,
farokhi2026precision}:
\begin{equation}\label{eq:functional_privacy_definition}
\E_w>0,\qquad \E_{\mathbf v}=0 \quad\forall,\mathbf v\neq0\text{ with }\mathbf v\perp\mathbf w.
\end{equation}
Thus, no independent linear combination orthogonal to the target can be estimated from the same statistical model.

Let $\mathcal R=\operatorname{range}(\mathbf Q)$. Because $\E_w>0$, one has $\mathbf w\in\mathcal R$. If $\dim\mathcal R\ge2$, there exists an $\mathbf x\in\mathcal R$ linearly independent of $\mathbf w$, and hence
\begin{equation}
    \mathbf v = \mathbf x- \frac{\mathbf w^T\mathbf x}{\mathbf w^T\mathbf w}\mathbf w \in\mathcal R\cap\mathbf w^\perp, \qquad \mathbf v\neq0.
\end{equation}
This would make $\mathbf v$ estimable and contradict \cref{eq:functional_privacy_definition}. Functional privacy therefore requires
\begin{equation}\label{eq:functional_range_sm}
    \operatorname{range}(\mathbf Q)=\operatorname{span}(\mathbf w).
\end{equation}
Since $\mathbf Q$ is symmetric and positive semidefinite, this is equivalent to
\begin{equation}\label{eq:functional_qfim_sm}
\mathbf Q=\kappa\mathbf w\mathbf w^T, \qquad \kappa>0.
\end{equation}
Conversely, \cref{eq:functional_qfim_sm} clearly satisfies \cref{eq:functional_privacy_definition}; it is therefore both necessary and sufficient for functional privacy.

\subsection{Relation Between Intrinsic Local and Functional Privacy}

For a nonlocalized target, i.e., a target supported on at least two nodes, functional privacy implies intrinsic local privacy. Indeed, $\operatorname{range}(\mathbf Q)=\operatorname{span}(\mathbf w)$, while $\mathbf e_i\notin\operatorname{span}(\mathbf w)$ for every $i$. The range-space conditions in \cref{eq:privacy_range_sm} are therefore satisfied. A localized target with $\mathbf w\parallel\mathbf e_i$ cannot possess intrinsic local privacy because estimability of $\mathbf w$ would simultaneously make $\mathbf e_i$ estimable.

This implication is less immediate if privacy is described using projected QFI. For the functionally private QFIM in \cref{eq:functional_qfim_sm}, the projected QFI along a participating local coordinate is
\begin{equation}\label{eq:local_projected_private}
    \mathbf e_i^T\mathbf Q\mathbf e_i =\kappa w_i^2=q_i>0, \qquad w_i\neq0.
\end{equation}
Thus, if all other parameters were known, changing $\theta_i$ would alter the authorized function and would remain detectable. Intrinsic local privacy instead asks whether $\theta_i$ can be identified independently when the other parameters are unknown. Since $\mathbf e_i\notin\operatorname{range}(\mathbf Q)$, one obtains
\begin{equation}
    q_i>0,\qquad \E_i=0,
\end{equation}
for every participating node. This distinction demonstrates why projected single-parameter sensitivity does not imply nuisance-aware local estimability.

The converse does not hold because intrinsic local privacy permits $\mathrm{rank}(\mathbf{Q})>1$; explicit higher-rank qubit constructions are given in Sec. VII. 

Finally, with prescribed diagonal capacities, the rank-one condition \cref{eq:functional_qfim_sm} is feasible only if $q_i=\kappa w_i^2$ for all $i$. Hence, all nonzero-weight nodes must have the same weighted capacity, $\frac{q_i}{w_i^2}=\kappa$ for all $\omega_i\ne 0$, while every zero-weight node must satisfy $q_i=0$. The resulting global EQFI is
\begin{equation}\label{eq:functional_bottleneck_sm}
\E_w=\kappa=\min_{i:w_i\neq0}\frac{q_i}{w_i^2},
\end{equation}
which saturates the fundamental bottleneck bound. Functional privacy therefore realizes the rank-one endpoint of intrinsic local privacy: all estimable information is concentrated into the single authorized global direction.

\subsection{Comparison with Previous Work}

Ref.~\cite{farokhi2026precision} recently established a complementary QFI duality for distributed sensing. For locally phase-encoded $N$-qubit probes and any two orthonormal sensing directions $\hat{\mathbf w}\perp\hat{\mathbf v}$, it gives
\begin{equation}
    F_{\hat{\mathbf w}}+F_{\hat{\mathbf v}} \leq \operatorname{Tr}(\mathbf Q)\leq N,
\end{equation}
where $F_{\hat{\mathbf u}}=\hat{\mathbf u}^T\mathbf Q\hat{\mathbf u}$ is the projected QFI and $\hat{\mathbf u} = \mathbf u/ \sqrt{\mathbf u^T \mathbf u}$ for all weights $\mathbf u$ satisfies $\sum_{i}|u_i| = 1$. This bound applies to arbitrary probes within the considered local-qubit encoding model; all equatorial probes saturate it for $N=2$, while GHZ probes provide a saturating family for $N\geq2$. In particular, $F_{\hat{\mathbf w}}=N$ exhausts the available trace budget and forces $F_{\hat{\mathbf v}}=0$ for every $\hat{\mathbf v}\perp\hat{\mathbf w}$, thereby connecting Heisenberg-limited directional sensitivity with functional privacy.

The present framework addresses a distinct, nuisance-aware estimation problem. Projected QFI quantifies sensitivity when the parameter vector is varied only along a prescribed direction. By contrast, when all local parameters are independently unknown, the precision for estimating $\theta_w=\mathbf w^T\bm{\theta}$ is governed by
\begin{equation}
\E_w=
\begin{cases}
(\mathbf w^T\mathbf Q^+\mathbf w)^{-1},
& \mathbf w\in\operatorname{range}(\mathbf Q),\\
0, & \mathbf w\notin\operatorname{range}(\mathbf Q),
\end{cases}
\end{equation}
which incorporates the penalty from all nuisance directions and remains well defined for singular QFIMs. For fixed diagonal local QFIs, our node--subnetwork phase map characterizes this effective information for an arbitrary positive-semidefinite QFIM, resolving the optimal, break-even, overcorrelated, and matched-singular regimes rather than constraining only a pair of orthogonal projected sensitivities.

The connection with previous privacy results is recovered as a special case. Functional privacy, as formulated by Hassani et al.~\cite{hassani2025privacy}, requires
$\operatorname{range}(\mathbf Q)=\operatorname{span}{(\mathbf w)}$, or equivalently
$\mathbf Q=\kappa\mathbf w\mathbf w^T$. For every participating node, this rank-one case lies at the perfectly aligned matched-singular endpoint of our phase map,
\begin{equation}
r_i\longrightarrow1,\qquad
|\mu_i|\longrightarrow1,\qquad
s_i-\mu_i r_i\longrightarrow0.
\end{equation}
Our intrinsic local-privacy condition is more general: it requires only
$\mathbf e_i\notin\operatorname{range}(\mathbf Q)$ for every $i$ while
$\mathbf w\in\operatorname{range}(\mathbf Q)$, and therefore permits higher-rank QFIMs that retain additional estimable collective modes. The distinction between projected QFI and nuisance-aware EQFI is illustrated explicitly in the preceding permutation-invariant example in Sec.~V.

\section{Qubit-State Constructions for Intrinsic Local Privacy}\label{sec:qubit_privacy_construction}

We consider qubit sensors with local sensing generators $H_i=Z_i/2$. For a pure probe $|\psi\rangle$, the QFIM is $Q_{ij}=\langle Z_iZ_j\rangle-\langle Z_i\rangle\langle Z_j\rangle$. 

\subsection{Physical principle and general construction}

The physical principle underlying intrinsic local privacy is to \textit{freeze selected collective sensing modes while retaining fluctuations along the target mode}. For a parameter-space direction $\mathbf a$, define the collective generator $G_{\mathbf a}=\sum_i a_iZ_i$. Its QFI is $\mathbf a^T\mathbf Q\mathbf a=\operatorname{Var}_{\psi}(G_{\mathbf a})$. Thus, if $|\psi\rangle$ is an eigenstate of $G_{\mathbf a}$, the parameter combination $\mathbf a^T\bm{\theta}$ produces only an unobservable global phase and belongs to the QFIM kernel.

To make every individual parameter inaccessible, we freeze a subspace of collective modes that involves every node. At the same time, these frozen modes must be orthogonal to the target so that the target can remain estimable. Let
\begin{equation}
    \mathcal K\subseteq\mathbf w^\perp, \qquad \dim\mathcal K=r_K,
\end{equation}
and require
\begin{equation}
    \mathbf e_i\not\perp\mathcal K \qquad\forall i.
    \label{eq:K_full_support}
\end{equation}
The latter condition means that every node participates in at least one frozen collective mode.

Choose a basis $\{\mathbf u^{(a)}\}_{a=1}^{r_K}$ of $\mathcal K$ and define
\begin{equation}
    G_a=\sum_{i=1}^{N}u_i^{(a)}Z_i,
    \qquad a=1,\ldots,r_K.
    \label{eq:privacy_generators}
\end{equation}
Because all $G_a$ are diagonal in the computational basis, they commute
and possess the joint spectral decomposition
\begin{equation}
    \mathcal H=\bigoplus_{\mathbf g}\mathcal H_{\mathbf g},
    \qquad
    \mathcal H_{\mathbf g}=\bigcap_{a=1}^{r_K}\operatorname{Eig}(G_a,g_a),
    \label{eq:joint_spectral_decomposition}
\end{equation}
where
\begin{equation}
    \mathcal H_{\mathbf g}=\operatorname{span}
    \left\{|\mathbf z\rangle:(\mathbf u^{(a)})^T\mathbf z=g_a, \forall a\right\}.
    \label{eq:joint_eigenspace_z}
\end{equation}
Here, $\mathbf z\in\{\pm1\}^N$ denotes the vector of local $Z$
eigenvalues associated with $|\mathbf z\rangle$.

We prepare a coherent probe within one joint eigenspace,
\begin{equation}
    |\psi\rangle=\sum_{\mathbf z\in\mathcal S_\psi} c_{\mathbf z}|\mathbf z\rangle,\qquad
    \mathcal S_\psi\subseteq\left\{\mathbf z:(\mathbf u^{(a)})^T\mathbf z=g_a, \forall a\right\},
    \label{eq:general_private_state}
\end{equation}
where $c_{\mathbf z}\neq0$ for every occupied configuration. Since
\begin{equation}
    G_a|\psi\rangle=g_a|\psi\rangle,\quad  \forall a
\end{equation}
we have
\begin{equation}
    (\mathbf u^{(a)})^T\mathbf Q\mathbf u^{(a)}=\operatorname{Var}_{\psi}(G_a)=0,\quad  \Rightarrow\quad  \mathcal K\subseteq\ker(\mathbf Q).
\end{equation}
Because every $\mathbf e_i$ has a nonzero component along
$\mathcal K$, this inclusion is already sufficient to make every local
parameter inaccessible:
\begin{equation}
    \mathbf e_i\notin\operatorname{range}(\mathbf Q),
    \qquad
    \mathcal E_i=0
    \quad\forall i.
\end{equation}

It remains necessary to ensure that the global target is not frozen.
For this purpose, define the variation subspace generated by the
occupied configurations,
\begin{equation}
    \mathcal V_\psi
    \equiv
    \operatorname{span}
    \left\{
    \mathbf z-\mathbf z':
    \mathbf z,\mathbf z'\in\mathcal S_\psi
    \right\}.
    \label{eq:variation_subspace}
\end{equation}
A direction $\mathbf a$ has zero QFI if and only if
$\mathbf a^T\mathbf z$ is identical for every occupied configuration.
Consequently,
\begin{equation}
    \ker(\mathbf Q)=\mathcal V_\psi^\perp, \qquad
    \operatorname{range}(\mathbf Q)=\mathcal V_\psi.
    \label{eq:kernel_from_variation_space}
\end{equation}
The target therefore remains estimable precisely when
\begin{equation}
    \mathbf w\in\mathcal V_\psi.
    \label{eq:target_retained_condition}
\end{equation}

Combining these observations, a sufficient and operational
construction of intrinsic local privacy is
\begin{equation}
    \mathcal K\subseteq\mathbf w^\perp,\qquad
    \mathbf e_i\not\perp\mathcal K\ \forall i,\qquad
    \mathbf w\in\mathcal V_\psi.
    \label{eq:physical_privacy_construction}
\end{equation}
The first two conditions freeze collective modes that involve every
node, whereas the third ensures that the target still fluctuates across
the occupied configurations. Together they imply
\begin{equation}
    \mathcal E_i=0\quad\forall i,
    \qquad
    \mathcal E_w>0.
\end{equation}

The actual QFIM kernel may be larger than the initially chosen
$\mathcal K$. This does not spoil intrinsic local privacy provided all
additional null directions remain orthogonal to $\mathbf w$. 

The physical construction is summarized by
\begin{equation}
    \begin{aligned}
    \mathbf w
    &\longrightarrow
    \mathcal K\subseteq\mathbf w^\perp\\
    &\longrightarrow
    \{G_a\}_{a=1}^{r_K}\\
    &\longrightarrow
    \mathcal H_{\mathbf g}\\
    &\longrightarrow
    |\psi\rangle\in\mathcal H_{\mathbf g}\\
    &\longrightarrow
    \mathcal K\subseteq
    \ker(\mathbf Q)=\mathcal V_\psi^\perp,
    \qquad
    \mathbf w\in\mathcal V_\psi.
    \end{aligned}
    \label{eq:privacy_construction_chain}
\end{equation}

Not every abstract choice of $\mathcal K$ is necessarily realizable
using fixed qubit generators $Z_i$, because the corresponding joint
eigenspace may not contain sufficiently many computational-basis
configurations. The examples below avoid this issue by explicitly
constructing the occupied configurations and their associated kernel.

\subsection{Uniform-weight construction}
\label{sec:uniform_weight_privacy_construction}

We now apply the general construction chain in \cref{eq:privacy_construction_chain} to the uniform target
\begin{equation}
    \mathbf w=\frac{\mathbf1_N}{N}.
\end{equation}
To avoid confusion with the phase-map overlap parameters $r_i$, we denote the desired QFIM nullity by
\begin{equation}
    r_{\mathrm K}\equiv\dim\ker(\mathbf Q),
    \qquad
    r\equiv\operatorname{rank}(\mathbf Q)=\dim\operatorname{range}(\mathbf Q)=N-r_{\mathrm K}.
\end{equation}
Here, $r$ is the desired dimension of the estimable subspace. For
$N\geq4$, the construction below realizes every
$r_{\mathrm K}=1,\ldots,N-1$, i.e. $r=1,\ldots,N-1$, while keeping the local capacities fixed. 

\paragraph*{1. Choice of the frozen subspace.---}

Following the first step of
\cref{eq:privacy_construction_chain}, we choose an
$r_{\mathrm K}$-dimensional subspace
$\mathcal K_{r_{\mathrm K}}\subseteq\mathbf w^\perp$. Specifically, we
define
\begin{equation}
    \mathcal K_{r_{\mathrm K}}=\left\{\mathbf k\in\mathbb R^N:\mathbf1_N^T\mathbf k=0,\quad k_j+k_N=0\ \text{for }j=1,\ldots,r-1\right\}.
    \label{eq:uniform_kernel}
\end{equation}
The first constraint ensures
$\mathcal K_{r_{\mathrm K}}\subseteq\mathbf w^\perp$, so the uniform
target is not frozen. The remaining $r-1$ constraints control the number of additional nonlocal modes that remain estimable. These $r$ constraints are independent, and therefore $ \dim\mathcal K_{r_{\mathrm K}} =N-r=r_{\mathrm K}$. For $N\geq4$, this subspace has a nonzero projection onto every coordinate direction, $\mathbf e_i\not\perp\mathcal K_{r_{\mathrm K}}$ for every $i$. 
Thus, freezing all modes in $\mathcal K_{r_{\mathrm K}}$ removes the estimability of every individual local parameter.

\paragraph*{2. Collective generators.---}

For each $\mathbf k\in\mathcal K_{r_{\mathrm K}}$, define the
collective generator
\begin{equation}
    G_{\mathbf k}=\sum_i k_iZ_i.
\end{equation}
Freezing the modes in $\mathcal K_{r_{\mathrm K}}$ only requires the probe to belong to a common eigenspace of these generators; the corresponding eigenvalues need not generally vanish. For the present uniform-weight construction, however, we choose their common zero-eigenvalue ($g_a=0$) subspace,
\begin{equation}
    \mathcal H_{\mathbf0}=\bigcap_{\mathbf k\in\mathcal K_{r_{\mathrm K}}}\ker G_{\mathbf k}.
    \label{eq:uniform_common_eigenspace}
\end{equation}
This choice allows each $Z$-configuration $|\mathbf v\rangle_Z$ to be paired with its complement $|-\mathbf v\rangle_Z$ within the same common eigenspace, because $G_{\mathbf k}|\pm\mathbf v\rangle_Z=\pm(\mathbf k^T\mathbf v)|\pm\mathbf v\rangle_Z=0$. These complementary pairs will ensure vanishing local means and identical local sensing capacities.

\paragraph*{3. Selection of configurations in the common eigenspace.---}

To construct a state in $\mathcal H_{\mathbf0}$, we need computational-basis configurations whose $Z$-eigenvalue vectors are orthogonal to every $\mathbf k\in\mathcal K_{r_{\mathrm K}}$. Equivalently, these configuration vectors must belong to $\mathcal K_{r_{\mathrm K}}^\perp$.

The first configuration is chosen as $\mathbf v^{(0)}=\mathbf1_N$. This choice is essential because $\mathbf w=\mathbf v^{(0)}/N$: including $\mathbf v^{(0)}$ ensures that the uniform target belongs to the fluctuating subspace and remains estimable.

We then choose
\begin{equation}
    \mathbf v^{(j)}=\mathbf1_N-2(\mathbf e_j+\mathbf e_N), \qquad j=1,\ldots,r-1.
    \label{eq:uniform_configuration_vectors}
\end{equation}
Each $\mathbf v^{(j)}$ is a physically allowed $Z$-eigenvalue
configuration obtained from $\mathbf1_N$ by reversing the outcomes at
nodes $j$ and $N$.

The simultaneous two-node reversal is important. If only node $j$
were reversed, the new configuration would be
$\mathbf1_N-2\mathbf e_j$, and $\mathbf1_N-\left(\mathbf1_N-2\mathbf e_j\right)=2\mathbf e_j$. 
The resulting variation subspace would contain the individual
direction $\mathbf e_j$, making $\theta_j$ estimable and destroying
local privacy. By reversing nodes $j$ and $N$ together, the new
direction becomes $\mathbf v^{(0)}-\mathbf v^{(j)}= 2(\mathbf e_j+\mathbf e_N)$, which is a collective two-node direction rather than a local one.

The selected vectors span
\begin{equation}
\begin{aligned}
    \operatorname{span} \{\mathbf v^{(0)},\ldots,\mathbf v^{(r-1)}\}
    &= \operatorname{span} \left\{ \mathbf1_N,\, \mathbf e_1+\mathbf e_N,\ldots, \mathbf e_{r-1}+\mathbf e_N \right\}\\
    &= \mathcal K_{r_{\mathrm K}}^\perp.
\end{aligned}
\label{eq:uniform_visible_subspace}
\end{equation}
Thus, these $r$ vectors form a basis of the desired estimable
subspace. In particular, they are linearly independent.

Moreover, for every $\mathbf k\in\mathcal K_{r_{\mathrm K}}$, $\mathbf k^T\mathbf v^{(j)}=0$ for all $j$. 
Therefore,
\begin{equation}
\begin{aligned}
    G_{\mathbf k}|\mathbf v^{(j)}\rangle_Z=0,\quad 
    G_{\mathbf k}|-\mathbf v^{(j)}\rangle_Z=0.
\end{aligned}
\label{eq:pair_same_common_eigenspace}
\end{equation}
Hence $|\mathbf v^{(j)}\rangle_Z$ and
$|-\mathbf v^{(j)}\rangle_Z$ belong to the same zero-eigenvalue
subspace of every frozen generator. This is the direct realization of
the common-eigenspace step in
\cref{eq:privacy_construction_chain}.

\paragraph*{4. Construction of the probe state.---}

For each observable direction $\mathbf v^{(j)}$, we coherently
superpose the two complementary configurations
$|\mathbf v^{(j)}\rangle_Z$ and
$|-\mathbf v^{(j)}\rangle_Z$. The complete probe is
\begin{equation}
    |\psi_{r_{\mathrm K}}^{(N)}\rangle = \frac{1}{\sqrt{2r}} \sum_{j=0}^{r-1} \left( |\mathbf v^{(j)}\rangle_Z + |-\mathbf v^{(j)}\rangle_Z \right), \qquad r=N-r_{\mathrm K}.
    \label{eq:uniform_private_family}
\end{equation}
Equation~\eqref{eq:pair_same_common_eigenspace} guarantees
\begin{equation}
    |\psi_{r_{\mathrm K}}^{(N)}\rangle \in\mathcal H_{\mathbf0}.
\end{equation}

The reason for including both $\mathbf v^{(j)}$ and $-\mathbf v^{(j)}$ is twofold. First, they belong to the same frozen joint eigenspace. Second, their difference is $\mathbf v^{(j)}-(-\mathbf v^{(j)}) =2\mathbf v^{(j)}$, so the complementary pair generates precisely the fluctuating direction $\mathbf v^{(j)}$. Each pair therefore opens one estimable collective mode without affecting the frozen subspace.

Equal weighting provides a canonical symmetric construction. Because each configuration appears together with its opposite, $\langle Z_i\rangle=0$. It also gives every selected collective mode equal metrological weight and keeps all diagonal capacities identical.

\paragraph*{5. QFIM kernel and intrinsic local privacy.---}

Since the local means vanish, the QFIM is
\begin{equation}
    \mathbf Q_{r_{\mathrm K}}^{(N)}=\frac{1}{r} \sum_{j=0}^{r-1} \mathbf v^{(j)}\mathbf v^{(j)T}. 
    \label{eq:uniform_private_QFIM}
\end{equation}
Its range is exactly the subspace generated by the selected
configurations:
\begin{equation}
    \operatorname{range}
    \bigl(\mathbf Q_{r_{\mathrm K}}^{(N)}\bigr) =
    \operatorname{span} \{\mathbf v^{(0)},\ldots,\mathbf v^{(r-1)}\}= \mathcal K_{r_{\mathrm K}}^\perp.
\end{equation}
Consequently,
\begin{equation}
    \ker \bigl(\mathbf Q_{r_{\mathrm K}}^{(N)}\bigr)= \mathcal K_{r_{\mathrm K}},
\end{equation}
and
\begin{equation}
    \operatorname{rank}\bigl(\mathbf Q_{r_{\mathrm K}}^{(N)}\bigr)=r,
    \qquad \dim\ker\bigl(\mathbf Q_{r_{\mathrm K}}^{(N)}\bigr)=r_{\mathrm K}.
\end{equation}

Because
\begin{equation}
    \mathbf w=\frac{\mathbf v^{(0)}}{N}\in\operatorname{range} \bigl(\mathbf Q_{r_{\mathrm K}}^{(N)}\bigr),
\end{equation}
the uniform target remains estimable. For $N\geq4$, no individual
coordinate vector belongs to the selected range:
\begin{equation}
    \mathbf e_i\notin\operatorname{range}\bigl(\mathbf Q_{r_{\mathrm K}}^{(N)}\bigr)\qquad\forall i.
\end{equation}
Hence
\begin{equation}
    \mathcal E_i=0\quad\forall i,\qquad \mathcal E_w>0.
\end{equation}

All nodes have the same local capacity,
\begin{equation}
    q_i=Q_{ii}=1,
\end{equation}
and the global EQFI is
\begin{equation}
    \mathcal E_w=\frac{N^2}{r}=\frac{N^2}{N-r_{\mathrm K}}.
    \label{eq:uniform_private_global_EQFI}
\end{equation}

At $r_{\mathrm K}=N-1$, only $\mathbf v^{(0)}=\mathbf1_N$ is retained and the probe reduces to the GHZ state. At $r_{\mathrm K}=1$, the single frozen mode is $\mathbf u=(-1,\ldots,-1,N-3,1)^T$, where the first $N-2$ entries are $-1$. Every component is nonzero for $N\geq4$, so this endpoint realizes the matched-singular boundary of the node--subnetwork phase map.

\subsection{Explicit four-node example}
\label{sec:N4_uniform_privacy_example}

The physical progression is particularly transparent for $N=4$ and
\begin{equation}
    \mathbf w=\frac14(1,1,1,1)^T.
\end{equation}
We begin with the GHZ pair, for which only the uniform mode fluctuates.
We then add complementary configuration pairs one at a time. Each new
pair opens one additional nonlocal sensing direction while preserving
the inaccessibility of every individual parameter.

The relevant configuration vectors are
\begin{equation}
\begin{aligned}
    \mathbf v^{(0)}&=(1,1,1,1)^T,\\
    \mathbf v^{(1)}&=(-1,1,1,-1)^T,\\
    \mathbf v^{(2)}&=(1,-1,1,-1)^T.
\end{aligned}
\label{eq:N4_uniform_vectors}
\end{equation}
Together with their opposite vectors, they correspond to
\begin{equation}
\begin{aligned}
    \pm\mathbf v^{(0)}
    &\longleftrightarrow
    |0000\rangle,\ |1111\rangle,\\
    \pm\mathbf v^{(1)}
    &\longleftrightarrow
    |1001\rangle,\ |0110\rangle,\\
    \pm\mathbf v^{(2)}
    &\longleftrightarrow
    |0101\rangle,\ |1010\rangle.
\end{aligned}
\end{equation}

\subsubsection{$r=1$: only the target mode fluctuates}

For $r=1$ and $r_{\mathrm K}=3$, the probe is
\begin{equation}
    |\psi_3\rangle=\frac{|0000\rangle+|1111\rangle}{\sqrt2}.
\end{equation}
Its QFIM is
\begin{equation}
    \mathbf Q_3=
    \begin{pmatrix}
        1&1&1&1\\
        1&1&1&1\\
        1&1&1&1\\
        1&1&1&1
    \end{pmatrix},
\end{equation}
with
\begin{equation}
    \operatorname{rank}(\mathbf Q_3)=1, \qquad \ker(\mathbf Q_3)=\mathbf1_4^\perp.
\end{equation}
Only the uniform target mode fluctuates, yielding
\begin{equation}
    \mathcal E_w=16.
\end{equation}
This saturates the bottleneck value
$q_i/w_i^2=1/(1/4)^2=16$.

\subsubsection{$r=2$: one additional nonlocal mode}

Adding the complementary pair
$|1001\rangle,|0110\rangle$ gives
\begin{equation}
|\psi_2\rangle=\frac{|0000\rangle+|1111\rangle+|1001\rangle+|0110\rangle}{2}.
\end{equation}
Its QFIM is
\begin{equation}
    \mathbf Q_2=
    \begin{pmatrix}
        1&0&0&1\\
        0&1&1&0\\
        0&1&1&0\\
        1&0&0&1
    \end{pmatrix},
\end{equation}
with
\begin{equation}
    \ker(\mathbf Q_2)=\operatorname{span}\left\{(1,0,0,-1)^T,\,(0,1,-1,0)^T\right\}.
\end{equation}
The target remains estimable, while every coordinate direction still
has a nonzero component along the frozen subspace. Hence
\begin{equation}
    \mathcal E_i=0\quad\forall i,\qquad \mathcal E_w=8.
\end{equation}

\subsubsection{$r=3$: minimally singular local privacy}

Adding the second complementary pair
$|0101\rangle,|1010\rangle$ gives
\begin{equation}
    |\psi_1\rangle=\frac{
    |0000\rangle+|1111\rangle+
    |1001\rangle+|0110\rangle+
    |0101\rangle+|1010\rangle
    }{\sqrt6}.
\end{equation}
Its QFIM is
\begin{equation}
    \mathbf Q_1=\frac{1}{3}
    \begin{pmatrix}
        3&-1&1&1\\
        -1&3&1&1\\
        1&1&3&-1\\
        1&1&-1&3
    \end{pmatrix},
\end{equation}
with
\begin{equation}
    \ker(\mathbf Q_1)=\operatorname{span}
    \left\{(-1,-1,1,1)^T\right\}.
\end{equation}
The remaining null mode involves every node and is orthogonal to the
uniform target. Therefore,
\begin{equation}
    \mathcal E_i=0\quad\forall i,
    \qquad
    \mathcal E_w=\frac{16}{3}.
\end{equation}
This is the minimally singular locally private state and is the member
directly captured by the matched-singular boundary of the
node--subnetwork phase map.

\subsubsection{Summary}

For all three probes,
\begin{equation}
    q_i=1,\qquad
    \mathcal E_i=0\quad\forall i.
\end{equation}
Their properties are summarized by
\begin{equation}
\begin{array}{c|c|c|c}
r_{\mathrm K}
&\operatorname{rank}(\mathbf Q) &\mathcal E_w &\text{occupied configurations}\\
\hline
3&1&16&2\\
2&2&8&4\\
1&3&16/3&6
\end{array}.
\end{equation}
This sequence directly illustrates the physical principle: the GHZ state freezes every mode except the target and therefore maximizes the global EQFI. Adding complementary configuration pairs restores additional nonlocal sensing modes, reducing the kernel dimension and redistributing information away from the target. Nevertheless, no individual coordinate direction becomes estimable, so intrinsic local privacy is preserved throughout.

\bibliography{references}